\documentclass[12pt]{article}

\usepackage[utf8]{inputenc}
\usepackage[T1]{fontenc}

\usepackage{a4wide}
\usepackage{amsfonts}
\usepackage{amssymb,amsmath,amsthm}
\usepackage{xspace}
\usepackage{tikz}
\usepackage{todo}
\usepackage{wrapfig}
\usepackage[shortcuts]{extdash}
\usepackage{hyperref}

\title{Unambiguous languages exhaust the index hierarchy\footnote{The author has been supported by Poland's National Science Centre grant no.~2016/21/D/ST6/00491.}}

\author{Micha{\l} Skrzypczak}

\newcommand{\eqdef}{\stackrel{\text{def}}=}

\newcommand{\comment}[1]{}

\newcommand{\lex}{<_{\mathrm{lex}}}
\newcommand{\lexeq}{\leq_{\mathrm{lex}}}
\newcommand{\lexge}{>_{\mathrm{lex}}}
\newcommand{\lexgeq}{\geq_{\mathrm{lex}}}

\newcommand{\defPlayer}[1]{\ensuremath{{\pmb{#1}}}\xspace}

\newcommand{\eve}{\defPlayer{\exists}}
\newcommand{\adam}{\defPlayer{\forall}}

\newcommand{\restr}{{\upharpoonright}}

\newcommand{\ident}[2]{\newcommand{#1}{\ensuremath{\texorpdfstring{\mathrm{#2}}{#2}}\xspace}}

\newcommand{\dL}{\mathtt{\scriptstyle L}}
\newcommand{\dR}{\mathtt{\scriptstyle R}}

\renewcommand{\mod}[1]{\allowbreak\mkern10mu({\operator@font mod}\,\,#1)}

\newcommand{\fun}[3]{\ensuremath{#1\colon #2 \to #3}}
\newcommand{\parfun}[3]{\ensuremath{#1\colon #2 \rightharpoonup #3}}

\ident{\dom}{dom}
\ident{\rg}{rg}
\ident{\id}{id}
\ident{\sgn}{sgn}

\newcommand{\mathcalsym}[1]{\ensuremath{\mathcal{#1}}\xspace}

\newcommand{\w}{\omega}

\newcommand{\Aa}{\mathcalsym{A}}

\newcommand{\Cc}{\mathcalsym{C}}

\newcommand{\Gg}{\mathcalsym{G}}

\newcommand{\Rr}{\mathcalsym{R}}

\newcommand{\Uu}{\mathcalsym{U}}

\newcommand{\Ww}{\mathcalsym{W}}

\ident{\lang}{L}

\ident{\init}{I}

\ident{\trees}{Tr}
\ident{\der}{der}
\ident{\rank}{rank}
\ident{\shave}{shave}

\newcommand{\dD}{\mathtt{\scriptstyle D}}

\newcommand{\pI}{{\mathtt{\scriptstyle \mathbf{1}}}}
\newcommand{\pII}{{\mathtt{\scriptstyle \mathbf{2}}}}

\newcommand{\iKnow}[1]{{\color{black}{#1}}}

\newcommand{\rmin}{\iKnow{i}}
\newcommand{\rmid}{\iKnow{j}}
\newcommand{\rpro}{\iKnow{\ell}}
\newcommand{\rmax}{\iKnow{k}}

\newcommand{\sigA}{\iKnow{\sigma}}

\newcommand{\psig}{\iKnow{s}}

\newcommand{\ordA}{\iKnow{\theta}}

\newcommand{\verA}{\iKnow{v}}

\newcommand{\finA}{\iKnow{u}}
\newcommand{\finB}{\iKnow{w}}

\newcommand{\infA}{\iKnow{\alpha}}
\newcommand{\infB}{\iKnow{\beta}}

\newcommand{\plyA}{\iKnow{\Pi}}

\newcommand{\tl}{\iKnow{p}}
\newcommand{\tr}{\iKnow{s}}

\newcommand{\boldclass}[3]{\ensuremath{\mathbf{#1}^{#2}_{#3}}}

\newcommand{\asigma}[1]{\boldclass{\Sigma}{1}{#1}}
\newcommand{\api}[1]{\boldclass{\Pi}{1}{#1}}

\newcommand{\spri}[1]{[#1]}
\newcommand{\spla}[1]{{\langle#1\rangle}}
\newcommand{\splp}[1]{\spla{#1{+}}}

\newcommand{\sneg}{{\sim}}

\newcommand{\Alp}[2]{A_{(#1,#2)}}
\newcommand{\Ang}[2]{A_{(#1,#2)}^{{\sim}}}

\newcommand{\Angp}[2]{A_{(#1,#2)}^{{+}{\sim}}}

\newcommand{\W}[2]{\ensuremath{W_{{#1},{#2}}}\xspace}
\newcommand{\Win}[3]{\mathrm{W}_{#1,(#2,#3)}}
\newcommand{\Wng}[3]{\mathrm{W}_{#1, (#2,#3)}^{{\sim}}}
\newcommand{\lan}[3]{L_{#1, (#2,#3)}}

\newcommand{\twoheaddownarrow}{\mathrel{\rotatebox[origin=c]{-90}{$\twoheadrightarrow$}}}

\newcommand{\tikzEvalInt}[2]{\pgfmathparse{int(#2)}{\global\edef#1{\pgfmathresult}}}

\usetikzlibrary{positioning}
\usetikzlibrary{decorations.pathreplacing}
\usetikzlibrary{automata,positioning}
\usetikzlibrary{decorations.pathmorphing}
\usetikzlibrary{decorations.markings}
\usetikzlibrary{decorations}
\usetikzlibrary{arrows}
\usetikzlibrary{patterns}
\usetikzlibrary{calc}
\usetikzlibrary{shapes}
\usetikzlibrary{fadings,	shadings}

\tikzstyle{ubrace} = [draw, thick, decoration={brace, mirror, raise=0.0cm}, decorate,
    every node/.style={anchor=north, yshift=-0.1cm}]
\tikzstyle{rbrace} = [draw, thick, decoration={brace, mirror, raise=0.0cm}, decorate,
    every node/.style={anchor=west, xshift= 0.1cm}]

\tikzstyle{obrace} = [draw, thick, decoration={brace, raise=0.0cm}, decorate,
    every node/.style={anchor=south, yshift= 0.1cm}]
\tikzstyle{lbrace} = [draw, thick, decoration={brace, raise=0.0cm}, decorate,
    every node/.style={anchor=east, xshift=-0.1cm}]

\tikzstyle{flow}  = [inner sep=0cm, node distance=0cm and 0cm]

\tikzstyle{toL} = [anchor=mid east, flow]
\tikzstyle{toR} = [anchor=mid west, flow]
\tikzstyle{toC} = [align=center, anchor=mid, flow]

\tikzstyle{letter} = [toC, scale=1.0]
\tikzstyle{ustate} = [toC, scale=1.0]
\tikzstyle{bstate} = [toC, scale=1.0]
\tikzstyle{transs} = [->, shorten <=2pt, shorten >=3pt]
\tikzstyle{transA} = [transs, bend left=30]

\tikzstyle{trdots}=[draw, thick, loosely dotted]

\newtheorem{theorem}{Theorem}[section]
\newtheorem{definition}[theorem]{Definition}
\newtheorem{lemma}[theorem]{Lemma}
\newtheorem{remark}[theorem]{Remark}
\newtheorem{example}[theorem]{Example}
\newtheorem{corollary}[theorem]{Corollary}

\theoremstyle{plain}
\newtheorem*{rep@theorem}{\rep@title}
\newcommand{\newreptheorem}[2]{%
\newenvironment{rep#1}[1]{%
\def\rep@title{#2 \ref{##1}}%
\begin{rep@theorem}}%
{\end{rep@theorem}}}

\newreptheorem{theorem}{Theorem}
\newreptheorem{lemma}{Lemma}
\newreptheorem{proposition}{Proposition}
\newreptheorem{conjecture}{Conjecture}
\newreptheorem{claim}{Claim}
\newreptheorem{fact}{Fact}

\newtheorem{fact}[theorem]{Fact}
\newtheorem{proposition}[theorem]{Proposition}
\newtheorem{claim}[theorem]{Claim}

\newcommand{\quas}[1]{\Sigma^{{\scriptscriptstyle \bigstar}}_{#1}}

\newcommand{\EL}{(EL)\xspace}
\newcommand{\AL}{(AL)\xspace}
\newcommand{\EI}{(EI)\xspace}
\newcommand{\AI}{(AI)\xspace}
\newcommand{\Db}{(\pmb{$\twoheaddownarrow$})\xspace}
\newcommand{\Dl}{(\pmb{$\downarrow$}$\tl$)\xspace}
\newcommand{\Dr}{(\pmb{$\downarrow$}$\tr$)\xspace}

\newcommand{\downl}{Step\ensuremath{\exists}{}(\tl)}
\newcommand{\downr}{Step\ensuremath{\forall}{}(\tr)}

\ident{\act}{Active}
\ident{\suk}{Succ}

\newcommand{\myPar}[1]{\medskip\noindent {\bf #1$\ $}}

\begin{document}

\maketitle

\begin{abstract}
This work is a~study of the expressive power of unambiguity in the case of automata over infinite trees. An~automaton is called unambiguous if it has at most one accepting run on every input, the language of such an automaton is called an unambiguous language. It is known that not every regular language of infinite trees is unambiguous. Except that, very little is known about which regular tree languages are unambiguous.

This paper answers the question whether unambiguous languages are of bounded complexity among all regular tree languages. The notion of complexity is the canonical one, called the (parity or Rabin\=/Mostowski) index hierarchy. The answer is negative, as exhibited by a~family of examples of unambiguous languages that cannot be recognised by any alternating parity tree automata of bounded range of priorities.

Hardness of the given examples is based on the theory of signatures in parity games, previously studied by Walukiewicz. This theory is further developed here to construct canonical signatures. The technical core of the article is a~parity game that compares signatures of a~given pair of parity games (without an increase in the index).
\end{abstract}

\section{Introduction}
\label{sec:intro}

Non\=/determinism provides a~machine with a~very powerful ability to \emph{guess} its choices. Depending on the actual model, it might enhance the expressive power of the considered machines or, while preserving the class of recognised languages, make the machines more succinct or effective. All these benefits come at the cost of algorithmic difficulties when handling non\=/deterministic devices. This complexity motivates a~search of ways of restricting the power of non\=/determinism. One of the most natural among these restrictions is a~semantic notion called \emph{unambiguity}: a~non\=/deterministic machine is called \emph{unambiguous} if it has at most one accepting run on every input.

Unambiguity turns out to be very intriguing in the context of automata theory~\cite{colcombet_determinism}. In the classical case of finite words it does not enhance the expressive power of the automata, still it simplifies some decision problems~\cite{stearns_unambiguous}. The situation is even more interesting in the case of infinite trees: the language of infinite trees labelled $\{a,b\}$ containing a~letter $a$ cannot be recognised by any unambiguous parity automaton~\cite{niwinski_unambiguous,walukiewicz_choice}. This example makes the impression that very few regular languages of infinite trees are in fact unambiguous (i.e.~can be recognised by an unambiguous automaton). However, there is only a~couple of distinct examples of ambiguous regular tree languages~\cite{bilkowski_unambiguity}. Our understanding of how many (or which) regular tree languages are unambiguous is far from being complete, in particular it is not known how to decide if a~given regular tree language is unambiguous.

Another way of understanding the power of unambiguous tree languages is aimed at estimating their descriptive complexity. The complexity can be measured either in terms of the topological complexity or of the parity index, i.e.~the range of priorities needed for an alternating parity tree automaton to recognise a~given language. Initially, it was considered plausible that all unambiguous tree languages are co-analytic ($\api{1}$); that is topologically not more complex than deterministic ones. Hummel in~\cite{hummel_gandalf} gave an example of an unambiguous language that is $\asigma{1}$-complete, in particular not $\api{1}$. Further improvements~\cite{duparc_unambiguous_index,hummel_phd} showed that unambiguous languages reach high into the second level of the index hierarchy. However, the question whether this is an upper bound on their index complexity was left open. In this paper we prove that it is not the case, as expressed by the following theorem.

\begin{theorem}
\label{thm:main}
For every $\rmin<\rmax$ there exists an unambiguous tree language $L$ that cannot be recognised by any alternating parity tree automaton (\emph{ATA}) that uses priorities $\{\rmin,\ldots,\rmax\}$. In other words, $L$ does not belong to the level $(\rmin,\rmax)$ of the index hierarchy.
\end{theorem}

The canonical examples of languages lying high in the index hierarchy~\cite{bradfield_simplifying,arnold_strict} are the languages $\W{\rmin}{\rmax}$ dating back to~\cite{jutla_determinacy,walukiewicz_signatures} (see e.g.~the formulae $W_n$ in~\cite{arnold_strict}). Unfortunately, the languages $\W{\rmin}{\rmax}$ are not unambiguous---one can interpret the choice problem~\cite{walukiewicz_choice} in such a~way that witnessing unambiguously that $t\in\W{1}{2}$ would indicate an MSO\=/definable choice function~\cite{shelah_choice,loding_choice}. Therefore, to prove Theorem~\ref{thm:main} we will use the following corollary of~\cite{niwinski_strict}.

\begin{corollary}[\cite{arnold_strict,niwinski_strict}]
\label{cor:reduction}
Let $L$ be a~set of trees. If there is a~continuous function $f$ s.t.~$\W{\rmin}{\rmax}=f^{-1}(L)$ then $L$ cannot be recognised by an ATA of index $(\rmin{+}1,\rmax{+}1)$.
\end{corollary}
Our aim is to enrich in a~continuous way a~given tree $t$ with some additional information denoted $f(t)$, such that an unambiguous automaton reading $f(t)$ can verify if $t\in\W{\rmin}{\rmax}$. Although this method is based on the topological concept of a~continuous mapping $f$, the construction provided in this paper is purely combinatorial; the core is a~definition of a~parity game $\Cc_P$ that compares the difficulty of a~given pair of parity games.


\section{Basic notions}
\label{sec:basic}

We use $\finA\cdot\finB$ to represent the concatenation of the two sequences. The symbol ${\preceq}$ stands for the prefix order. By $\w=\{0,1,\ldots\}$ we denote the set of natural numbers.

A~\emph{(ranked) alphabet} is a~non\=/empty finite set $A$ of \emph{letters} where each letter $a\in A$ comes with its own finite \emph{arity}. A~\emph{tree} over an alphabet $A$ is a~partial function $\parfun{t}{\w^\ast}{A}$ where the domain $\dom(t)$ is non\=/empty, prefix\=/closed, and if $\finA\in\dom(t)$ is a~\emph{node} with a~$k$ary letter $a= t(\finA)$ then $\finA\cdot l\in\dom(t)$ if and only if $l<k$, i.e.~$\finA$ (the \emph{father}) has \emph{children} $\finA\cdot 0$, $\finA\cdot 1$, \ldots, $\finA\cdot (k{-}1)$. The set of all trees over $A$ is denoted $\trees_A$. The element $\epsilon\in\dom(t)$ is called the \emph{root} of $t$. A~\emph{branch} of a~tree $t$ is a~sequence $\infA\in \w^\w$ such that for all $n\in\w$ the finite prefix $\infA\restr n$ is a~node of $t$. It is easy to encode ranked alphabets using alphabets of fixed arity (or even binary), however for the technical simplicity we will work with ranked ones here.

If $\finA\in\dom(t)$ is a~node of a~tree $t$ then by $t\restr \finA$ we denote the tree $\finB\mapsto t(\finA\cdot \finB)$ with the domain $\{\finB\mid \finA\cdot \finB\in\dom(t)\}$. A~tree of the form $t\restr \finA$ for $\finA\in\dom(t)$ is called a~\emph{subtree} of $t$ in~$\finA$. If $a\in A$ is a~$k$ary letter and $t_0,\ldots,t_{k-1}$ are trees then by $a\big(t_0,\ldots,t_{k-1}\big)$ we denote the unique tree $t$ with $t(\epsilon)=a$ and $t\restr(i)=t_i$ for $i=0,\ldots,k{-}1$.

If $X\subseteq \dom(t)$ is a~set of nodes of a~tree $t$ and $\finA\in\dom(t)$ then by $X\restr\finA$ we denote the set $\{\finB\mid\finA\cdot\finB\in X\}$, which is a~subset of $\dom(t\restr\finA)$.

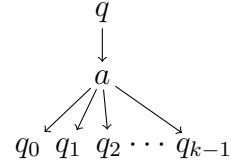
\begin{wrapfigure}{r}{5.5cm}
\centering
\begin{tikzpicture}[scale=0.9]
\node[ustate] (un) at ( 1.1, +0.0) {$q$};
\node[letter] (ln) at ( 1.1, -1.0) {$a$};

\draw[transs] (un) -- (ln);

\node[bstate] (bn) at (0.0, -2) {$q_0$};
\draw[transs] (ln) -- (bn);
\node[bstate] (bn) at (0.6, -2) {$q_1$};
\draw[transs] (ln) -- (bn);
\node[bstate] (bn) at (1.2, -2) {$q_2$};
\draw[transs] (ln) -- (bn);

\node[bstate] (bn) at (1.8, -2) {$\cdots$};

\node[bstate] (bn) at (2.6, -2) {$q_{k-1}$};
\draw[transs] (ln) -- (bn);
\end{tikzpicture}
\caption{A~representation of a~transition $(q,a,q_0,\ldots,q_{k-1})$, for a~$k$ary letter~$a$.}
\label{fig:transition}
\end{wrapfigure}
\myPar{Automata} A~\emph{non\=/deterministic parity tree automaton} is a~tuple $\Aa=\langle A, Q, \Delta, I, \Omega\rangle$, where $A$ is a~ranked alphabet; $Q$ is a~finite set of \emph{states}; $\Delta$ is a~finite set of \emph{transitions}---tuples of the form $(q,a,q_0,\ldots,q_{k-1})$ where $a\in A$ is a~$k$ary letter and $q,q_0,\ldots,q_{k-1}$ are states; $I\subseteq Q$ is a~set of \emph{initial states}; and $\fun{\Omega}{Q}{\w}$ is a~\emph{priority} mapping.

A~\emph{run} of an automaton $\Aa$ over a~tree $t$ over the alphabet~$A$ is a~function $\fun{\rho}{\dom(t)}{Q}$ such that $\rho(\epsilon)\in I$ and for every $\finA\in \dom(t)$ with a~$k$ary letter $a= t(\finA)$, the tuple $\big(\rho(\finA),a,\rho(\finA\cdot 0),\ldots,\rho(\finA\cdot (k{-}1))\big)$ is a~transition in $\Delta$. A~run~$\rho$ is \emph{accepting} if for every branch $\infA$ of $t$ the lowest priority of the states appearing infinitely many times along $\infA$ (i.e.~$\liminf_{n\to\infty}\Omega\big(\rho(\infA\restr n)\big)$) is even. An~automaton $\Aa$ \emph{accepts} a~tree $t$ if there exists an accepting run of $\Aa$ over $t$. The language of an automaton $\Aa$ (denoted $\lang(\Aa)$) is the set of trees accepted by $\Aa$. A~set of trees over an~alphabet $A$ is called \emph{regular} if it is recognised by a~non\=/deterministic parity tree automaton. For a~detailed introduction to the theory of automata over infinite trees, see~\cite{thomas_languages}.

An automaton $\Aa$ is \emph{unambiguous} if for every tree $t$ there exists at most one accepting run of $\Aa$ over $t$. An automaton $\Aa$ is \emph{deterministic} if $I=\{q_\init\}$ is a~singleton and for every $q\in Q$ and $k$ary letter $a\in A$ it has at most one transition of the form $(q,a,q_0,\ldots,q_{k-1})$ in $\Delta$. A~language is \emph{unambiguous} (resp. \emph{deterministic}) if it can be recognised by an unambiguous (resp. deterministic) automaton. Clearly each deterministic automaton is unambiguous but the converse is not true. Due to~\cite{walukiewicz_choice} we know that there are regular tree languages that are ambiguous (i.e.~not unambiguous).

\myPar{Games} A~\emph{game} with players $\pI$ and $\pII$ is a~tuple $\Gg=\langle V, E, v_\init, W\rangle$ where: $V=V_{\pI}\sqcup V_{\pII}$ is a~set of \emph{positions} split into the \emph{$\pI$\=/positions} $V_{\pI}$ and \emph{$\pII$\=/positions} $V_{\pII}$; $E\subseteq V\times V$ is a~set of \emph{edges}; $v_\init\in V$ is an \emph{initial position}; and $W\subseteq V^\w$ is a~\emph{winning condition}. We will denote by $P$ the players, i.e.~$P\in\{\pI,\pII\}$, $\bar{P}$ is the opponent of $P$. For $\verA\in V$ by $\verA\cdot E$ we denote the set of successors $\{\verA'\mid (\verA,\verA')\in E\}$. We assume that for each $\verA\in V$ the set $\verA\cdot E$ is non\=/empty. A~non\=/empty finite or infinite sequence $\plyA\in V^{\leq \w}$ is a~\emph{play} if $\plyA(0)=\verA_\init$ and for each $0<i<|\plyA|$ there is an edge $\big(\plyA(i{-}1),\plyA(i)\big)$. Notice that if $\langle V, E\rangle$ is a~tree then there is an equivalence between finite plays and positions $\verA\in V$. An infinite play $\plyA$ is \emph{winning for $\pI$} if $\plyA\in W$; otherwise $\plyA$ is winning for $\pII$.

A~non\=/empty and prefix\=/closed set of plays $\Sigma$ with no ${\preceq}$\=/maximal element (i.e.~no leaf) is called a~\emph{behaviour}. We call a~behaviour \emph{$P$\=/full} if for every play $(\verA_0,\ldots,\verA_n)\in\Sigma$ with $\verA_n\in V_{P}$ and all $\verA'\in \verA\cdot E$ we have $(\verA_0,\ldots,\verA_n,\verA')\in \Sigma$. We call a~behaviour \emph{$P$\=/deterministic} if for every play $(\verA_0,\ldots,\verA_n)\in\Sigma$ with $\verA_n\in V_{P}$ there is a~unique $\verA'\in \verA\cdot E$ such that $(\verA_0,\ldots,\verA_n,\verA')\in \Sigma$.

A~\emph{quasi\=/strategy} of a~player $P$ is a~behaviour that is $\bar{P}$\=/full. A~\emph{strategy} of $P$ is a~quasi\=/strategy of $P$ that is $P$\=/deterministic. A~quasi\=/strategy is \emph{positional} if the fact whether a~play $(\verA_0,\ldots,\verA_n, \verA_{n+1})$ belongs to $\Sigma$ depends only on $\verA_n$.

A~\emph{partial strategy} of $P$ is a~$P$\=/deterministic behaviour---it defines the unique choices of $P$ but may not respond to some choices of $\bar{P}$. We say that a~play $(\verA_0,\ldots,\verA_n, \verA_{n+1})\notin\Sigma$ is \emph{not reachable} by a~partial strategy $\Sigma$ if $(\verA_0,\ldots,\verA_n)\in \Sigma$ and $\verA_n\in V_{\bar{P}}$. If $\Sigma$ is a~(partial) strategy of $P$ and $(\verA_0,\ldots,\verA_n,\verA')\in\Sigma$ with $\verA_n\in V_P$ then we say that $\Sigma$ \emph{moves to} $\verA'$ in $(\verA_0,\ldots,\verA_n)$.

A~strategy $\Sigma$ of $P$ is \emph{winning} if every \emph{infinite play of $\Sigma$} (i.e.~$\plyA$ such that $\forall n\in\w.\ \plyA\restr n\in\Sigma$) is winning for $P$. A~game is (\emph{positionally}) \emph{determined} if one of the players has a~(positional) winning strategy. We say that a~position $\verA$ of a~game $\Gg$ is \emph{winning for $P$} (resp. \emph{losing for $P$}) if $P$ (resp. $\bar{P}$) has a~winning strategy in the game $\Gg$ with $\verA_\init:= \verA$.

\myPar{Topology} In this work we use only basic notions of descriptive set theory and topology, see~\cite{kechris_descriptive,thomas_topology} for a~broader introduction. The space $\trees_A$ with the product topology is homeomorphic to the Cantor space. One can take as the basis of this topology the sets of the form $\{t\in\trees_A\mid t(\finA_1){=}a_1, t(\finA_2){=}a_2,\ldots, t(\finA_n){=}a_n\}$ for finite sequences $(\finA_1$, $a_1$, $\finA_2$, $a_2, \ldots, \finA_n$, $a_n)$. The open sets in $\trees_A$ are obtained as unions of basic open sets. A~function $\fun{f}{X}{Y}$ is continuous if the pre\=/image of each basic open set in $Y$ is open in $X$.

\section{The languages}
\label{sec:language}

Let us fix a~pair of numbers $\rmin < \rmax$. Our aim is to encode a~general parity game with players $\pI$ and $\pII$ and priorities $\{\rmin,\ldots,\rmax\}$ as a~tree over a~fixed ranked alphabet $\Alp{\rmin}{\rmax}$. That alphabet consists of: unary symbols $\spri{\rmin},\spri{\rmin{+}1},\ldots, \spri{\rmax}$ indicating priorities of positions; and binary symbols $\spla{\pI}$ and $\spla{\pII}$ which leave the choice of the subtree to the respective player.

The game induced by a~tree $t\in\trees_{\Alp{\rmin}{\rmax}}$ is denoted $\Gg(t)$. The set of positions of $\Gg(t)$ is $\dom(t)$ and the edge relation contains pairs father---child. The initial position is $\epsilon$ and a~position $\verA\in\dom(t)$ is a~$\pI$\=/position iff $t(\verA)=\spla{\pI}$. An infinite play of that game is won by $\pI$ if and only if the minimal priority $\rmid$ that occurs infinitely often during the play is even\footnote{We restrict our attention to the trees in which every second symbol on each branch is a~unary symbol representing a~priority (i.e.~$\spri{\rmid}$ for $\rmid\in\{\rmin,\ldots,\rmax\}$); every tree can implicitly be transformed into that format by padding with the maximal priority $\rmax$ (such a~padding does not influence the winner of $\Gg(t)$).}. Since the graph of $\Gg(t)$ is a~tree, we identify finite plays in $\Gg(t)$ with positions $\verA\in\dom(t)$. Therefore, (quasi / partial) strategies in $\Gg(t)$ can be seen as specific subsets $\Sigma\subseteq\dom(t)$ and infinite plays of these strategies as branches of $t$. 

For a~player $P\in\{\pI,\pII\}$ the language $\Win{P}{\rmin}{\rmax}$ contains a~tree $t$ if $P$ has a~winning strategy in $\Gg(t)$. It is easy to see that $\Win{\pI}{\rmin}{\rmax}$ is homeomorphic (i.e.~topologically equivalent) to $\W{\rmin}{\rmax}$ from~\cite{walukiewicz_signatures} (the case of $P=\pII$ is dual, we put $\W{\rmin+1}{\rmax+1}$ then).

As it turns out, the languages $\Win{P}{\rmin}{\rmax}$ are not expressive enough to allow enrichment of a~tree $t$ into $f(t)$
. To enlarge their expressive power we will extend the alphabet with a~unary symbol $\sneg$ that will correspond to a~swap of the players in $\Gg(t)$. The enhanced alphabet will be denoted $\Ang{\rmin}{\rmax}$. We will say that a~tree $t$ over the alphabet $\Ang{\rmin}{\rmax}$ is \emph{well\=/formed} if there is no branch with infinitely many symbols~$\sneg$.

Consider a~node $\finA\in\dom(t)$ in a~well\=/formed tree $t$ over $\Ang{\rmin}{\rmax}$. We will say that $\finA$ is \emph{switched} if there is an odd number of nodes $\finB\prec \finA$ such that $t(\finB)=\sneg$. Otherwise $\finA$ is \emph{kept}. These notions represent the fact that each symbol $\sneg$ swaps the players in $\Gg(t)$.

If $t\in\trees_{\Ang{\rmin}{\rmax}}$ is well\=/formed then the game $\Gg(t)$ is well\=/defined in a similar way as before. Formally, the set of positions $V$, the initial position $v_\init$, and the edge relation $E$ are defined in the same way as in the case of $t$ over $\Alp{\rmin}{\rmax}$. It remains to define $V_{\pI}$, $V_{\pII}$, and $W$, taking into account the swapping symbol $\sneg$.

Let a~position $\verA\in V$ belong to $V_{\pI}$ if either $\verA$ is \emph{kept} and $t(\verA)=\spla{\pI}$ or $\verA$ is \emph{switched} and $t(\verA)=\spla{\pII}$. The priority of a~position $\verA$ such that $t(\verA)=\spri{\rmid}$ is $\rmid$ if $\verA$ is \emph{kept} and $\rmid+1$ otherwise. An infinite play of $\Gg(t)$ is winning for $\pI$ if the least priority that occurs infinitely many times is even. The language $\Wng{P}{\rmin}{\rmax}$ contains a~well\=/formed tree $t$ over $\Ang{\rmin}{\rmax}$ if $P$ has a~winning strategy in $\Gg(t)$.

\begin{remark}
By the assumption of well\=/formedness, if $\plyA$ is an infinite play of $\Gg(t)$ (i.e.~a branch of $t$) then for sufficiently big $n$ either all the nodes $\plyA\restr n$ are \emph{switched} or all the nodes $\plyA\restr n$ are \emph{kept}. Thus, we can call the play $\plyA$ accordingly as \emph{switched} or \emph{kept}.

The parity condition $W$ of $\Gg(t)$ contains an infinite play $\plyA$ if one of the cases holds:
\begin{itemize}
\item $\plyA$ is \emph{kept} and the minimal number $\rmid$ such that for infinitely many $n$ we have $t(\plyA\restr n)=\spri{\rmid}$ is even,
\item $\plyA$ is \emph{switched} and the minimal number $\rmid$ such that for infinitely many $n$ we have $t(\plyA\restr n)=\spri{\rmid}$ is odd.
\end{itemize}
\end{remark}

Since the games $\Gg(t)$ have a parity winning condition and therefore are determined (see~\cite{jutla_determinacy,mostowski_parity_games}), we obtain:

\begin{fact}
If $t\in\trees_{\Ang{\rmin}{\rmax}}$ is well\=/formed then the following equivalence holds:
\[t\in\Wng{P}{\rmin}{\rmax}\quad\text{ iff }\quad t\notin\Wng{\bar{P}}{\rmin}{\rmax}\quad\text{ iff }\quad\sneg(t)\in\Wng{\bar{P}}{\rmin}{\rmax}.\]
%
\end{fact}

The additional information added by $f$ will be kept under third children of new ternary variants of the symbols $\spla{\pI}$ and $\spla{\pII}$, denoted $\splp{\pI}$ and $\splp{\pII}$. Let the alphabet $\Angp{\rmin}{\rmax}$ be $\Ang{\rmin}{\rmax}$ where instead of the symbols $\spla{\pI}$ and $\spla{\pII}$ we have $\splp{\pI}$ and $\splp{\pII}$ respectively.

Consider a~tree $r$ over the extended alphabet $\Angp{\rmin}{\rmax}$. By $\shave(r)$ we denote the tree over the non\=/extended alphabet $\Ang{\rmin}{\rmax}$, where instead of each subtree of the form $\splp{\pI}(t_\dL,t_\dR,t_2)$ one puts the subtree $\spla{\pI}(t_\dL,t_\dR)$; the same for $\splp{\pII}$ and $\spla{\pII}$. More formally, the function $\fun{\shave}{\trees_{\Angp{\rmin}{\rmax}}}{\trees_{\Ang{\rmin}{\rmax}}}$ is defined recursively as:
\begin{align*}
\shave\big(\splp{P}(t_\dL,t_\dR,t_2)\big)&=\spla{P}\Big(\shave(t_\dL),\shave(t_\dR)\Big)&\text{for $P\in\{\pI,\pII\}$,}\\
\shave\big(\sneg(t)\big)&=\sneg\big(\shave(t)\big),\\
\shave\big(\spri{\rmid}(t)\big)&=\spri{\rmid}\big(\shave(t)\big)&~\text{for $\rmid\in\{\rmin,\ldots,\rmax\}$.}
\end{align*}

Notice that $\dom(\shave(r))\subseteq\dom(r)$ and the labels of $\shave(r)$ correspond to the labels of $r$ in the respective nodes (up to the additional ${+}$ in $r$). We will say that a~tree $r$ over the alphabet $\Angp{\rmin}{\rmax}$ is \emph{well\=/formed} if for every its subtree $r'=r\restr \finA$ the tree $\shave(r')$ is well\=/formed in the standard sense. In other words, $r$ is well\=/formed if there is no branch of $r$ that contains infinitely many symbols $\sneg$ but only finitely many directions $2$ (the direction $2$ corresponds to moving outside $\shave(r')$).

We are now in position to define the witnesses proving Theorem~\ref{thm:main}. For that we will define an unambiguous automaton $\Uu$ recognising certain language of trees over the alphabet $\Angp{\rmin}{\rmax}$. In this section we will prove that $\Uu$ is unambiguous. In the rest of the article we show that $\Uu$ (with a~restricted set of initial states) recognises a~language high in the index hierarchy.

\begin{definition}
The set of states of $\Uu$ is $\{0,\ldots,\rmax{+}2\}\times\{\pI,\pII\}\times\{\dD,\dL,\dR\}$. Let $\Omega(\rmid,P,d)=\rmid$. The transitions of $\Uu$ are depicted in Figure~\ref{fig:full-strat}. The set of initial states of $\Uu$ contains all the states of the form $(0,\star,\star)$ (recall that $\star$ represents all the possible choices).
\end{definition}

Intuitively, the first coordinate of a~state $q$ of $\Uu$ is its priority; the second coordinate is the winner for $\Gg(\shave(r'))$ for the current subtree~$r'$; while the third coordinate indicates the actual strategy if there is ambiguity and $\dD$ otherwise. The transition over $\sneg$ represents that $\sneg$ swaps the players; the next two transitions correspond to positions that are not controlled by a~(claimed) winner $P$ over a~given subtree; and the last two transitions correspond to a~position that is controlled by $P$. In the lower two transitions the choice of a~direction $\dL$ or $\dR$ depends on the declared winner $P$ in the third child of the current node.


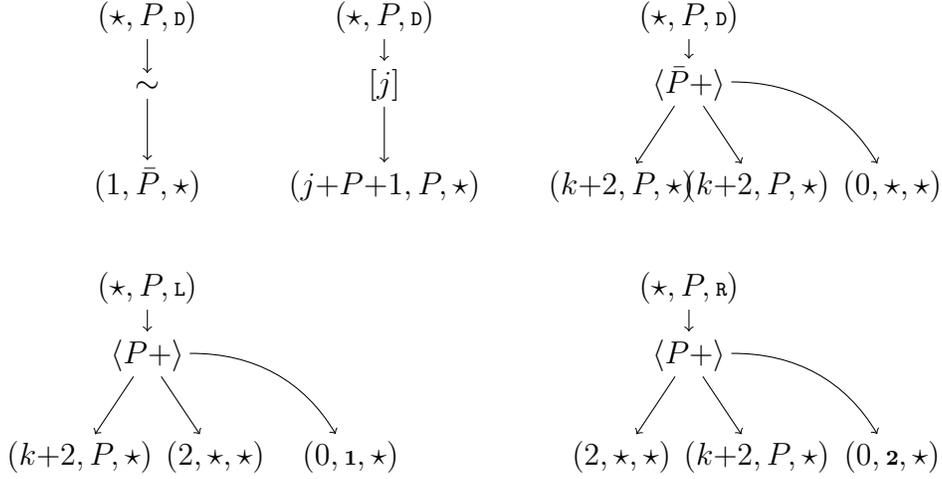
\begin{figure}
\centering
\begin{tikzpicture}[scale=0.9]
\node[ustate] (un) at ( 1.0, +1) {$(\star,P,\dD)$};
\node[letter] (ln) at ( 1.0, -0) {$\sneg$};
\node[bstate] (bn) at ( 1.0, -1.5) {$(1,\bar{P},\star)$};
\draw[transs] (un) -- (ln);
\draw[transs] (ln) -- (bn);

\node[ustate] (up) at ( 4.5, +1) {$(\star,P,\dD)$};
\node[letter] (lp) at ( 4.5, -0) {$\spri{\rmid}$};
\node[bstate] (bp) at ( 4.5, -1.5) {$(\rmid{+}P{+}1,P,\star)$};
\draw[transs] (up) -- (lp);
\draw[transs] (lp) -- (bp);

\node[ustate] (uO) at ( 9.0, +1) {$(\star,P,\dD)$};
\node[letter] (lO) at ( 9.0, -0) {$\splp{\bar{P}}$};
\node[bstate] (LO) at ( 8.0, -1.5) {$(\rmax{+}2,P,\star)$};
\node[bstate] (RO) at (10.0, -1.5) {$(\rmax{+}2,P,\star)$};
\node[bstate] (tO) at (12.0, -1.5) {$(0,\star,\star)$};
\draw[transs] (uO) -- (lO);
\draw[transs] (lO) -- (LO);
\draw[transs] (lO) -- (RO);
\draw[transA] (lO) edge[transA] (tO);

\node[ustate] (uO) at ( 1.0, -3.0) {$(\star,P,\dL)$};
\node[letter] (lO) at ( 1.0, -4.0) {$\splp{P}$};
\node[bstate] (LO) at ( 0.0, -5.5) {$(\rmax{+}2,P,\star)$};
\node[bstate] (RO) at ( 2.0, -5.5) {$(2,\star,\star)$};
\node[bstate] (tO) at ( 4.0, -5.5) {$(0,\pI,\star)$};
\draw[transs] (uO) -- (lO);
\draw[transs] (lO) -- (LO);
\draw[transs] (lO) -- (RO);
\draw[transA] (lO) edge[transA] (tO);

\node[ustate] (uO) at ( 9.0, -3.0) {$(\star,P,\dR)$};
\node[letter] (lO) at ( 9.0, -4.0) {$\splp{P}$};
\node[bstate] (LO) at ( 8.0, -5.5) {$(2,\star,\star)$};
\node[bstate] (RO) at (10.0, -5.5) {$(\rmax{+}2,P,\star)$};
\node[bstate] (tO) at (12.0, -5.5) {$(0,\pII,\star)$};
\draw[transs] (uO) -- (lO);
\draw[transs] (lO) -- (LO);
\draw[transs] (lO) -- (RO);
\draw[transA] (lO) edge[transA] (tO);
\end{tikzpicture}
\caption{The transitions of the automaton $\Uu$, where $P\in\{\pI,\pII\}$ stands for a~player; $\rmid\in\{\rmin,\ldots,\rmax\}$ is a~priority; and $\star$ represents all the possible choices on a~given coordinate.}
\label{fig:full-strat}
\end{figure}

Consider a~run $\rho$ of $\Uu$ over a~tree $r$, let $\finA\in\dom(r)$, and assume that $\rho(\finA)$ is of the form $(\star,P,\star)$. In that case one can extract from the third coordinates of $\rho$ a~strategy $\Sigma$ of $P$ in $\Gg(\shave(r\restr\finA))$ that will be called the \emph{$\rho$-strategy in $\finA$}. This strategy is defined inductively, preserving the invariant that for each $\finB\in\Sigma$ the node $\finB$ is \emph{kept} in $\shave(r\restr\finA)$ if and only if $\rho(\finA\cdot\finB)$ is of the form $(\star,P,\star)$. We start with $\Sigma$ containing the initial position $\epsilon$. Now consider a~position $\finB$ in $\Sigma$. If $\finB$ is controlled by $P$ (i.e.~$r(\finA\cdot\finB)=\splp{P}$ for $\finB$ \emph{kept} and $\splp{\bar{P}}$ for $\finB$ \emph{switched}) then the strategy $\Sigma$ moves to the position $\finB\cdot 0$ (resp. $\finB\cdot 1$) if the state $\rho(\finA\cdot\finB)=(\star,\star,d)$ satisfies $d=\dL$ (resp. $d=\dR$). In the positions $\finB\in\Sigma$ not controlled by $P$ the strategy $\Sigma$ has no choice and contains all the children of $\finB$ in $\shave(r\restr\finA)$. It is easy to check that the transitions of $\Uu$ guarantee that the invariant is preserved.

\newcommand{\lemCorrectStrategy}{
Let $\rho$ be a~run of $\Uu$ over $r$. Then $\rho$ is accepting if and only if $r$ is well\=/formed and for every $\finA\in\dom(r)$ the $\rho$-strategy in $\finA$ is winning in $\Gg\big(\shave(r\restr\finA)\big)$.
}

\begin{lemma}
\label{lem:correct-strategy}
\lemCorrectStrategy
\end{lemma}

\begin{proof}
First assume that $r$ is not well\=/formed and consider a~run $\rho$ of $\Uu$ over $r$. Since $r$ is not well\=/formed we know that $\shave(r\restr\finA)$ is not well\=/formed for some $\finA\in\dom(r)$. Let $\infA$ be a~branch of $\shave(r\restr\finA)$ that contains infinitely many occurrences of $\sneg$. Since $\infA$ is a~branch of $\shave(r\restr\finA)$, $\infA\in\{0,1\}^\w$ and therefore for all non\=/empty $\finB\prec \infA$ the priority of the state $\rho(\finA\cdot \finB)$ is non\=/zero. By the assumption about the occurrences of $\sneg$, for infinitely many $\finB\prec \infA$ the priority of the state $\rho(\finA\cdot \finB)$ is $1$. Therefore, the run $\rho$ is not accepting.

Now consider a~well\=/formed tree $r$ and a~run $\rho$ such that all the $\rho$-strategies are winning for the respective players. We need to show that $\rho$ is accepting. Consider a~branch $\infA$ of $r$.

First, consider the case that $\infA$ contains infinitely many directions $2$ (i.e.~infinitely many times it moves to the third child of a~ternary node). In that case the priority $0$ appears infinitely many times in $\rho$ on $\infA$ and therefore the parity condition is satisfied on that branch.

The opposite case is that $\infA$ contains only finitely many times $2$. Since $r$ is well\=/formed the symbol $\sneg$ must appear only finitely many times on $\infA$ in $r$. Therefore, the priorities $0$ and~$1$ appear only finitely many times in $\rho$ on $\infA$. If the priority $2$ appears infinitely many times in $\rho$ on $\infA$ then the parity condition is satisfied on that branch. Therefore, assume contrarily that the priority $2$ appears only finitely many times in $\rho$ on $\infA$.

Let $\finA\prec \infA$ be a~node such that all the occurrences of the direction $2$ in $\infA$, of the symbol~$\sneg$ in $r$, and of the priority $2$ in $\rho$ appear in the nodes that are prefixes of $\finA$. Let $\infA = \finA\cdot \infB$. Assume that $\rho(\finA)=(\star,P,\star)$. By the assumption about the $\rho$\=/strategies, we know that there exists a~strategy $\Sigma$ of $P$ that is winning in $\Gg\big(\shave(r\restr\finA)\big)$. Moreover, we know that $\infB$ is a~branch of $\shave(r\restr \finA)$.

Because the priority $2$ is not used by $\rho$ on $\infA$ below $\finA$, $\infB$ is an infinite play of $\Sigma$. As $\sneg$ does not appear in $\shave(r\restr\finA)$ all the positions along $\infB$ are \emph{kept}. Since $\Sigma$ is winning, the play following $\infB$ needs to be winning for $P$. Notice that the priorities of the states of $\rho$ on $\infA$ below $\finA$ are either $k+2$ (after reading $\splp{\pI}$ and $\splp{\pII}$) or $\rmid+P+1$ (after reading $\spri{\rmid}$). Thus, the assumption that $\infB$ is winning for $P$ in $\Gg\big(\shave(r\restr\finA)\big)$ implies that the minimal priority appearing infinitely many times in $\rho$ on $\infA$ is even.

What remains is to consider the case of $r$ with an accepting run $\rho$. As we have seen, $r$ must be well\=/formed in that case. By a~reasoning dual to the above one, each infinite play of a~$\rho$\=/strategy must be winning for the declared player because of the way the priorities of $\Uu$ are computed.
\end{proof}

The following fact follows directly from the above lemma.

\begin{fact}
\label{ft:full-to-win}
Assume that $\rho$ is an accepting run of $\Uu$ over a~tree $r$, $\finA\in\dom(r)$, and $\rho(\finA)$ is of the form $(\star,P,\star)$. Then $P$ wins $\Gg\big(\shave(r\restr \finA)\big)$. In particular, if $\rho$ and $\rho'$ are two accepting runs of $\Uu$ over the same tree $r$ then the second coordinates of $\rho$ and $\rho'$ are equal.
\end{fact}

\begin{lemma}
\label{lem:second-to-other}
If $\rho$ and $\rho'$ are two runs of $\Uu$ (possibly not accepting) over a~tree $r$ and the second coordinates of $\rho$ and $\rho'$ are equal then $\rho=\rho'$.
\end{lemma}

\begin{proof}
The third coordinate of a~run in $\finA\in\dom(r)$ depends on $r(\finA)$, the second coordinate of the run in $\finA$, and (if one of the lower two transitions from Figure~\ref{fig:full-strat}
is used) the second coordinate of the run in $\finA\cdot 2$. Thus, the third coordinates of $\rho$ and $\rho'$ must agree.

The first coordinates of $\rho$ and $\rho'$ in the root are $0$. Consider a~node $\finA$ and its child $\finA'\in\dom(r)$. The first coordinate of a~run in $\finA'$ depends on $r(\finA)$ and the last two coordinates of the run in $\finA$. Therefore, also the first coordinates of $\rho$ and $\rho'$ must agree.
\end{proof}

\begin{definition}
\label{def:languages}
Take $P\in\{\pI,\pII\}$ and let $\lan{P}{\rmin}{\rmax}$ be the language recognised by the automaton~$\Uu$ with the set of initial states restricted to the states of the form $(0,P,\star)$.
\end{definition}

Fact~\ref{ft:full-to-win} together with Lemma~\ref{lem:second-to-other} imply the following corollary.

\begin{corollary}
The languages $\lan{P}{\rmin}{\rmax}$ are unambiguous.
\end{corollary}

Thus, to complete the proof of Theorem~\ref{thm:main}, one needs to prove that the languages $\lan{P}{\rmin}{\rmax}$ climb up the index hierarchy. This will be done using Corollary~\ref{cor:reduction} in Section~\ref{sec:hardness}. However, before we move to that section, we need to develop the theory of \emph{signatures}.

\section{Signatures}
\label{sec:signatures}

The technical core of this work is based on a concept allowing to compare how much a~player $P$ \emph{prefers} one tree over another. This is achieved by assigning to each tree its \emph{$P$\=/signature} and proving that moving to a~subtree with an optimal signature guarantees that the player wins whenever possible. The theory of signatures comes from~\cite{streett_signatures} and results of Walukiewicz, e.g.~\cite{walukiewicz_signatures_pushdown}. The notion of signatures used here is more demanding than the classical ones, as we require equalities instead of inequalities in the invariants from Lemma~\ref{lem:sig-exist}.

In this section we define the signatures and prove their fundamental properties, further sections of the article use these notions to prove hardness of the languages $\lan{P}{\rmin}{\rmax}$.

We say that a~number $\rmid\in\w$ is \emph{$P$\=/losing} if $\rmin\leq\rmid\leq\rmax$ and $\rmid$ is odd (resp. even) if $P=\pI$ (resp. $P=\pII$). A~number $\rmid\in\{\rmin,\ldots,\rmax\}$ that is not $P$\=/losing is called \emph{$P$\=/winning}. A~$P$\=/signature is either $\infty$ or a~tuple of countable ordinals $(\ordA_{\rmin'},\ordA_{\rmin'+2},\ldots,\ordA_{\rmax'})$, indexed by $P$\=/losing numbers---$\rmin'$ is the minimal and $\rmax'$ is the maximal $P$\=/losing number. $P$\=/signatures that are not $\infty$ are well\=/ordered by the lexicographic order ${\lexeq}$ in which the ordinals with smaller indices are more important. Let $\infty$ be the maximal element of~${\lexeq}$.

A~$\pI$\=/signature $\sigA_{\pI}(t)=(42)$ with $\rmin=0$ and $\rmax=2$ means that the best what player $\pI$ can hope for is to visit at most $42$ times a~$\spri{1$}\=/node (possibly interleaved with nodes of priority~$2$) before the first $\spri{0}$\=/node is visited (if ever). After visiting a~$\spri{0}$\=/node, the $\pI$\=/signature of the subtree may grow, starting again a~counter of nodes of priority $1$ to be visited. The $\pI$\=/signature $(\omega)$ means that $\pII$ can choose a~finite number of $\spri{1}$-nodes that will be visited; however the choice needs to be done before the first such node is seen.

The following two lemmas express the crucial properties of the signatures.

\newcommand{\lemSigExist}[1]{
There exists a~unique point\=/wise minimal pair of assignments $t\mapsto\sigA_P(t)$ for $P\in\{\pI,\pII\}$ that assign to each well\=/formed tree $t$ over $\Ang{\rmin}{\rmax}$ a~$P$\=/signature $\sigA_P(t)$ such that:
\begin{enumerate}
\item $\sigA_P(t)=\infty$ if and only if $P$ loses $\Gg(t)$;
#1{it:sig-inf}
\item $\sigA_P\big(\sneg(t)\big)=(0,\ldots,0)$ if $P$ wins $\Gg\big(\sneg(t)\big)$ (i.e.~$\bar{P}$ wins $\Gg(t)$);
#1{it:sig-neg}
\item $\sigA_P\big(\spri{\rmid}(t)\big)=(\ordA_{\rmin'},\ldots,\ordA_{\rmid-1},0,0 \ldots, 0)\qquad$ if $\sigA_P(t)=(\ordA_{\rmin'},\ldots,\ordA_{\rmax'})$ and $\rmid$ is $P$\=/winning;
#1{it:sig-j-win}
\item $\sigA_P\big(\spri{\rmid}(t)\big)=(\ordA_{\rmin'},\ldots,\ordA_{\rmid-2},\ordA_\rmid{+}1,0, \ldots, 0)$ if $\sigA_P(t)=(\ordA_{\rmin'},\ldots,\ordA_{\rmax'})$ and $\rmid$ is $P$\=/losing;
#1{it:sig-j-los}
\item $\sigA_P\big(\spla{P}(t_{\dL},t_{\dR})\big)=\min\big\{\sigA_P(t_{\dL}), \sigA_P(t_{\dR})\big\}$;
#1{it:sig-min}
\item $\sigA_P\big(\spla{\bar{P}}(t_{\dL},t_{\dR})\big)=\max\big\{\sigA_P(t_{\dL}), \sigA_P(t_{\dR})\big\}$.
#1{it:sig-max}
\end{enumerate}
}

\begin{lemma}
\label{lem:sig-exist}
\lemSigExist{\label}
\end{lemma}

Let us fix the functions $\sigA_P$ for $P\in\{\pI,\pII\}$ as above. Consider a~well\=/formed tree $t$ over the alphabet $\Ang{\rmin}{\rmax}$. We say that a~strategy $\Sigma$ of a~player~$P$ in $\Gg(t)$ is \emph{optimal} (or $\sigA$\=/optimal) if:
\begin{itemize}
\item In a~position $\finA\in\dom(t)$ that is \emph{kept} in $t$ and $t\restr \finA = \spla{P}\big(t_\dL,t_\dR\big)$ the strategy $\Sigma$ moves to a~subtree of a~minimal value of~$\sigA_P$; i.e.~$\Sigma$ can move to $\finA\cdot 0$ if $\sigA_P(t_\dL)\lexeq \sigA_P(t_\dR)$ and to $\finA\cdot 1$ if $\sigA_P(t_\dL)\lexgeq \sigA_P(t_\dR)$. If the values $\sigA_P(t_\dL)$ and $\sigA_P(t_\dR)$ are equal then $\Sigma$ can move in any of the two directions.
\item In a~position $\finA$ that is \emph{switched} in $t$ and $t\restr \finA = \spla{\bar{P}}\big(t_\dL,t_\dR\big)$ the strategy $\Sigma$ uses the same rule as above but uses the function $\sigA_{\bar{P}}$ to compare the $\bar{P}$\=/signatures of the subtrees.
\end{itemize}
Notice that according to the above definition there might be more than one optimal strategy.

\newcommand{\lemSigOptWin}{
If $t\in\Wng{P}{\rmin}{\rmax}$ and $\Sigma$ is an~optimal strategy of $P$ in $\Gg(t)$ then $\Sigma$ is winning.
}

\begin{lemma}
\label{lem:sig-opt-win}
\lemSigOptWin
\end{lemma}

The rest of this section is devoted to a construction of the functions $\sigA_P$ and to proving Lemmas~\ref{lem:sig-exist}, \ref{lem:sig-opt-win}. We start with a~construction of the functions $\sigA_P$ for $P\in\{\pI,\pII\}$ in Subsection~\ref{ssec:sig-construction}. In Subsection~\ref{ssec:sig-properties} we prove that the constructed functions satisfy the invariants in Lemma~\ref{lem:sig-exist} and we provide a~proof of an important Lemma~\ref{lem:sig-subopt}. Subsection~\ref{ssec:sig-consequences} proves a~statement stronger than Lemma~\ref{lem:sig-opt-win}. Finally, Subsection~\ref{ssec:sig-optimal} is devoted to a~proof that the constructed functions $\sigA_P$ are point\=/wise minimal, what concludes the proof of Lemma~\ref{lem:sig-exist}.

The theory follows the ideas in~\cite{walukiewicz_signatures}, however as we insist on having equalities instead of inequalities in Lemma~\ref{lem:sig-exist} (particularly in Item~\ref{it:sig-max}) the construction here is more complex than the standard constructions. As a~side effect, we get the point\=/wise minimality of the defined $P$\=/signatures (proved in Subsection~\ref{ssec:sig-optimal}), which is not needed for the rest of the theory but shows that the definitions given here are canonical.

\subsection{Construction of \texorpdfstring{$\sigA_P$}{sigP}}
\label{ssec:sig-construction}

We start with a~construction of concrete functions satisfying the invariants from Lemma~\ref{lem:sig-exist}.

Consider a~tree $t$ over $\Ang{\rmin}{\rmax}$ and a~winning strategy $\Sigma$ of a~player $P$ in $\Gg(t)$. A~position $\finA\in\Sigma$ is \emph{active} if $t(\finA)=\spri{\rmid}$ for a~$P$\=/losing number $\rmid$ and for every $\finB\prec \finA$ it is not the case that $t(\finB)=\spri{\rmid'}$ with $\rmid'<\rmid$ nor $t(\finB)=\sneg$. In other words, $\finA$ is active if a~$P$\=/losing priority $\rmid$ occurs in $\finA$ and no lower priority nor the swapping symbol~$\sneg$ occurs on the path from the root to $\finA$. The set of active positions in $\Sigma$ is denoted $\act(\Sigma)$ (it might be empty). Notice that the set $\act(\Sigma)$ is a~set of nodes of $t$, each labelled by a~$P$\=/losing number. Moreover, the priorities of $t$ in $\act(\Sigma)$ are monotone:
\begin{equation}
\label{eq:active-monotone}
\text{if $\finA\preceq \finB$ with $\finA,\finB\in\act(\Sigma)$ then $t(\finA)\geq t(\finB)$},
\end{equation}
where the latter order is the order of priorities---natural numbers.

\begin{figure}
\centering
\begin{tikzpicture}[scale=0.99]
\tikzstyle{actt}=[draw, circle, inner sep=1pt]

\newcommand{\vvm}{0}
\newcommand{\vvu}{1.5}
\newcommand{\vvl}{-1}

\newcommand{\vvU}{2.0}
\newcommand{\vvM}{1.0}
\newcommand{\vvL}{-1}

\node[toL] at (-0.4,0) {$t=$};

\node[letter, actt] (m0) at ( 0, \vvm) {$\spri{5}$};
\node[letter      ] (m1) at ( 1, \vvm) {$\spri{4}$};
\node[letter      ] (m2) at ( 2, \vvm) {$\spri{5}$};
\node[letter, actt] (m3) at ( 3, \vvm) {$\spri{3}$};

\node[letter] at (4, \vvm+0.5) {$\finB$};
\node[letter      ] (m4) at ( 4, \vvm) {$\spla{\pII}$};

\node[letter      ] (u5) at ( 5, \vvu) {$\spla{\pII}$};
\node[letter      ] (l5) at ( 5, \vvl) {$\spri{4}$};

\node[letter] at (6, \vvU+0.5) {$\finA_1$};
\node[letter] at (6, \vvM-0.5) {$\finA_2$};
\node[letter] at (6, \vvL-0.5) {$\finA_3$};

\node[letter, actt] (u6) at ( 6, \vvU) {$\spri{3}$};
\node[letter      ] (m6) at ( 6, \vvM) {$\spri{2}$};
\node[letter, actt] (l6) at ( 6, \vvL) {$\spri{3}$};

\node[letter, actt] (u7) at ( 7, \vvU) {$\spri{1}$};
\node[letter      ] (m7) at ( 7, \vvM) {$\spri{3}$};
\node[letter, actt] (l7) at ( 7, \vvL) {$\spri{3}$};

\node[letter      ] (u8) at ( 8, \vvU) {$\spri{3}$};
\node[letter, actt] (m8) at ( 8, \vvM) {$\spri{1}$};
\node[letter, actt] (l8) at ( 8, \vvL) {$\spri{1}$};

\node[letter      ] (u9) at ( 9, \vvU) {$\spri{0}$};
\node[letter, actt] (m9) at ( 9, \vvM) {$\spri{1}$};
\node[letter      ] (l9) at ( 9, \vvL) {$\spri{4}$};

\node[letter      ] (u10) at (10, \vvU) {$\spri{5}$};
\node[letter      ] (m10) at (10, \vvM) {$\sneg$};
\node[letter      ] (l10) at (10, \vvL) {$\spri{2}$};

\node[letter      ] (u11) at (11, \vvU) {$\spri{3}$};
\node[letter      ] (m11) at (11, \vvM) {$\spri{1}$};
\node[letter      ] (l11) at (11, \vvL) {$\spri{3}$};

\node[letter      ] (u12) at (12, \vvU) {$\spri{0}$};
\node[letter      ] (m12) at (12, \vvM) {$\spri{0}$};
\node[letter      ] (l12) at (12, \vvL) {$\spri{0}$};

\node[letter      ] (u13) at (13, \vvU) {$t_\pI$};
\node[letter      ] (m13) at (13, \vvM) {$t_\pII$};
\node[letter      ] (l13) at (13, \vvL) {$t_\pI$};

\foreach \x in {0,...,3} {
	\tikzEvalInt{\y}{\x+1}
	\draw[transs] (m\x) -- (m\y);
}

\draw[transs] (m4) -- (u5);
\draw[transs] (m4) -- (l5);

\draw[transs] (u5) -- (u6);
\draw[transs] (u5) -- (m6);
\draw[transs] (l5) -- (l6);

\foreach \x in {6,...,12} {
	\tikzEvalInt{\y}{\x+1}
	\draw[transs] (u\x) -- (u\y);
	\draw[transs] (m\x) -- (m\y);
	\draw[transs] (l\x) -- (l\y);
}
\end{tikzpicture}
\caption{An example set $\act(\Sigma)$ for a~tree $t\in\Wng{\pI}{0}{5}$.}
\label{fig:active}
\end{figure}
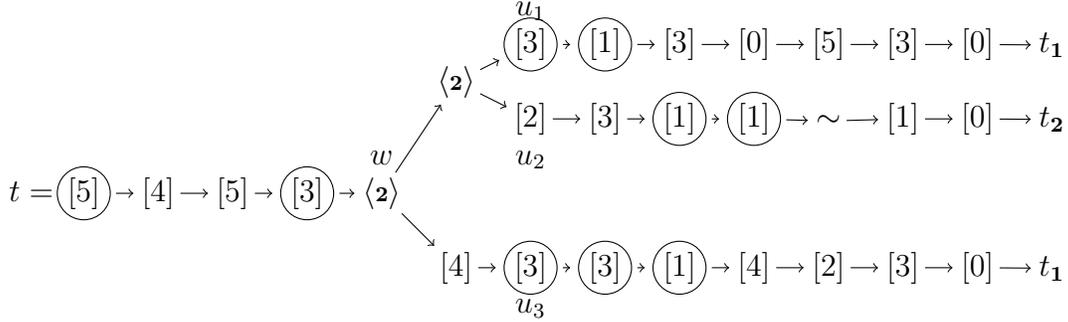

Recall that $\Sigma\restr\finA = \{\finB\mid \finA\cdot\finB\in\Sigma\}$. Therefore, if $\finA\in\Sigma$ then $\Sigma\restr\finA$ is a~strategy of $P$ in $t\restr\finA$.

The monotonicity of priorities in $\act(\Sigma)$ implies that if $\finA\in\act(\Sigma)$ then
\begin{equation}
\label{eq:active-restrictive}
\act(\Sigma)\restr\finA = \act(\Sigma\restr\finA).
\end{equation}

\begin{example}
Figure~\ref{fig:active} provides an illustration to the definition of the set $\act(\Sigma)$. We assume that $t_\pI=\spri{0}(t_\pI)$ and $t_\pII=\spri{1}(t_\pII)$ are two simple trees in $\Wng{\pI}{0}{5}$ and $\Wng{\pII}{0}{5}$ respectively. Then the tree $t$ belongs to $\Wng{\pI}{0}{5}$ and there exists only one strategy $\Sigma$ of $\pI$ in~$\Gg(t)$. The elements of the set $\act(\Sigma)$ are in circles. The values from Definition~\ref{def:signatures} for that tree are as follows:
\begin{align*}
\psig(\finA_1,\Sigma)&=(1,0,0),\\
\psig(\finA_2,\Sigma)&=(2,0,0),\\
\psig(\finA_3,\Sigma)&=(1,2,0),\\
\psig(\finB,\Sigma)&=(2,0,0),\\
\psig(\epsilon,\Sigma)&=(2,1,1),\\
\sigA_P(t)&=(2,1,1),\\
\sigA_P(t\restr\finA_1)&=(1,1,0),\\
\sigA_P(t\restr\finA_2)&=(2,0,0).
\end{align*}
\end{example}

Our constructions and proofs will follow the schema of well\=/founded induction over the set $\Sigma$ with the following transitive and irreflexive relation:
\[\finA\gg \finB\quad\text{if and only if}\quad \finA,\finB\in\Sigma\text{ and }\finA\prec\finB\text{ and }\exists{\finB'\in\act(\Sigma)}.\ \finA\preceq\finB'\preceq\finB.\]

\begin{lemma}
\label{lem:act-is-well-founded}
If $\Sigma$ is winning then the relation ${\gg}$ is well\=/founded (i.e.~it has no infinite descending chain).
\end{lemma}

\begin{proof}
An infinite descending chain in ${\gg}$ would indicate an infinite sequence $\finB_0\prec \finB_1\prec\ldots$ of elements of $\act(\Sigma)$. Such a~sequence must lie on an infinite branch $\infA$ in $\Sigma$. Let $\rmid$ be the minimal priority that occurs on nodes in $\act(\Sigma)$ on $\infA$. By the definition, no priority lower than $\rmid$ occurs on $\infA$ in $t$ (otherwise the considered infinitely many nodes in $\act(\Sigma)$ would not be active). Therefore, $\rmid$ is the minimal priority that occurs infinitely many times on $\infA$ in $t$. As $\rmid$ is $P$\=/losing, $\Sigma$ is not winning for $P$, a~contradiction.
\end{proof}

Consider a~node $\finA\in\Sigma$. Let $\suk_\finA(\Sigma)$ be the set of ${\preceq}$\=/successors of $\finA$ in $\act(\Sigma)$:
\[\suk_\finA(\Sigma)\eqdef\{\finB\in\act(\Sigma)\mid \finA\preceq\finB\wedge \lnot\exists {\finB'\in\act(\Sigma)}.\ \finA\preceq\finB'\prec\finB\}.\]
By the definition, if $\finB\in \suk_\finA(\Sigma)$ and $\finA\neq\finB$ then $\finA\gg\finB$. Similarly, if $\finA\prec \finB\in\Sigma$ and $\finA\in\act(\Sigma)$ then $\finA\gg\finB$. This means that the following definition follows the schema of well\=/founded induction over~${\gg}$, as in both cases the value $\psig(\finA,\Sigma)$ depends on the values of $\psig(\finA',\Sigma)$ with $\finA\gg\finA'$.

\begin{definition}
\label{def:signatures}
Consider the function $\psig(\cdot,\Sigma)$ defined by the following properties for $\finA\in\Sigma$:
\begin{enumerate}
\item If $\finA\in\act(\Sigma)$, $t(\finA)=\spri{\rmid}$, and $\finA'=\finA\cdot 0$ is the unique child of $\finA$ in $t$ then
\[\psig(\finA,\Sigma)=\big(\ordA_{\rmin'},\ldots,\ordA_{\rmid}+1,0,\ldots,0)\text{, where }\psig(\finA',\Sigma)=\big(\ordA_{\rmin'},\ldots,\ordA_{\rmax'}).\]
\label{it:inductive-psig-los}
\item If $\finA\notin\act(\Sigma)$ then
\[\psig(\finA,\Sigma)= \sup_{\finB\in \suk_\finA(\Sigma)} \psig(\finB,\Sigma).\]
\label{it:inductive-psig-rest}
\end{enumerate}
Moreover, for $P\in\{\pI,\pII\}$ and $t\in\trees_A$ put
\begin{align}
\sigA_P(t)&\eqdef \inf_{\text{$\Sigma$ winning for $P$}} \psig(\epsilon,\Sigma).
\label{eq:def-signature}
\end{align}
\end{definition}

Observe that as the values $\psig(\finB,\Sigma)$ are tuples of countable ordinals, the supremum in Item~\ref{it:inductive-psig-rest} is also a~tuple of countable ordinals (i.e.~a~$P$\=/signature). If $\finA\in\Sigma$ but $\finA\notin\act(\Sigma)$ is a ${\gg}$\=/minimal element then this supremum ranges over the empty set and equals $(0,\ldots,0)$.

\begin{lemma}
\label{lem:inductive-psig}
The function $\psig(\cdot,\Sigma)$ is monotone: if $\finA\preceq\finB\in\Sigma$ then $\psig(\finA,\Sigma)\lexgeq \psig(\finB,\Sigma)$.
\end{lemma}

\begin{proof}
Consider a~pair of nodes $\finA\prec\finB$. If $\finA\gg\finB$ then the monotonicity follows directly from the inductive definition: in each step of the definition we do not decrease the value of $\psig(\cdot,\Sigma)$. Otherwise, when $\finA\prec\finB$ but $\finA\not\gg\finB$ then $\finA,\finB\notin\act(\Sigma)$ and $\suk_\finB(\Sigma)\subseteq \suk_\finA(\Sigma)$ and we use monotonicity of the supremum operator.
\end{proof}

\begin{fact}
\label{ft:restriction}
If $\finA\in\act(\Sigma)$ then $\psig(\finA,\Sigma)=\psig(\epsilon,\Sigma\restr\finA)$.
\end{fact}

\begin{proof}
From~\eqref{eq:active-restrictive} we know that $\act(\Sigma)\restr\finA = \act(\Sigma\restr\finA)$. Since the definition of $\psig(\cdot,\Sigma)$ depends only on the set $\act(\Sigma)$ and its labelling by $t$, the values $\psig(\finA,\Sigma)$ and $\psig(\epsilon,\Sigma\restr\finA)$ must agree.
\end{proof}

Consider $\finA\in\Sigma$ and a~priority $\rmid\in\{\rmin,\ldots,\rmax\}$. Let 
\[\act_\finA(\Sigma,\rmid) =\{\finB\in\act(\Sigma)\mid \finA\preceq \finB\text{ and }t(\finB)\leq\rmid\}.\]
When we speak about $\act_\finB(\Sigma\restr\finA,\rmid)$ then we take $t\restr\finA$ instead of $t$ in the above definition, following the convention that we think about $\act(\Sigma)$ as a~set together with the labelling by~$t$.

\begin{lemma}
\label{lem:psig-restr}
If $\finA\in\Sigma$ and $\rmid$ is a~$P$\=/losing number then
\[\psig(\finA,\Sigma)\restr\rmid=\sup_{\finB\in \act_\finA(\Sigma,\rmid)} \big(\psig(\finB,\Sigma)\restr\rmid\big)=\sup_{\finB\in \act_\finA(\Sigma,\rmid)}\big(\psig(\epsilon,\Sigma\restr\finB)\restr\rmid\big).\]
\end{lemma}

\begin{proof}
We start with the first equality proved by an induction w.r.t. ${\gg}$. If $\finA\in\act(\Sigma)$ then either $t(\finA)>\rmid$ or $\finA\in\act_\finA(\Sigma,\rmid)$. In the latter case the equality is trivial by the monotonicity of $\psig(\cdot,\Sigma)$, see Lemma~\ref{lem:inductive-psig}. In the former case it follows from the inductive assumption for the unique child $\finA'$ of $\finA$ in $t$, as
\[\act_\finA(\Sigma,\rmid)=\act_{\finA'}(\Sigma,\rmid),\]
and the formula in Item~\ref{it:inductive-psig-los} of Definition~\ref{def:signatures} preserves $\psig(\finA,\Sigma)\restr\rmid$.

If $\finA\notin\act(\Sigma)$ then the equality holds by combining the supremum in Item~\ref{it:inductive-psig-rest} of Definition~\ref{def:signatures} with the one ranging over $\act_\finA(\Sigma,\rmid)$:
\[\act_\finA(\Sigma,\rmid)=\bigcup_{\finB\in\suk_\finA(\Sigma)}\act_\finB(\Sigma,\rmid).\]

The second equality follows immediately from Fact~\ref{ft:restriction}.
\end{proof}

\begin{lemma}
\label{lem:sig-invariant}
If $\finA\in\Sigma$, $\rmid$ is a~$P$\=/losing number, and for every $\finB\prec \finA$ the value $t(\finB)$ is not~$\sneg$ nor a~priority lower than~$\rmid$ then $\psig(\epsilon,\Sigma\restr\finA)\restr\rmid = \psig(\finA,\Sigma)\restr\rmid$.
\end{lemma}

\begin{proof}
We start from observing that under our assumptions\footnote{Notice that~\eqref{eq:active-restrictive} does not apply because $\finA$ might be not active.}:
\[\act(\Sigma)\restr\finA=\act(\Sigma\restr\finA).\]
The ${\subseteq}$ containment holds always as if $\finA\cdot\finB$ is active in $\Sigma$ then $\finB$ is active in $\Sigma\restr\finA$. The opposite direction follows from the assumption about the labels of $t$ on the path to $\finA$. Therefore, we immediately get that
\begin{equation}
\label{eq:active-equal-restr}
\act_\finA(\Sigma,\rmid)\restr\finA=\act_\epsilon(\Sigma\restr\finA,\rmid),
\end{equation}
because the difference between $\act(\Sigma)\restr\finA$ and $\act_\finA(\Sigma,\rmid)\restr\finA$ depends only on the labels of $t$ in the respective nodes.

Therefore, by applying Lemma~\ref{lem:psig-restr} twice, we get that:
\begin{align*}
\psig(\finA,\Sigma)\restr\rmid &= \sup_{\finA\cdot\finB\in \act_\finA(\Sigma,\rmid)}\Big(\psig\big(\epsilon,\Sigma\restr(\finA\cdot\finB)\big)\restr\rmid\Big)=\\
&=\sup_{\finB\in \act_\epsilon(\Sigma\restr\finA,\rmid)}\Big(\psig\big(\epsilon,\Sigma\restr(\finA\cdot\finB)\big)\restr\rmid\Big)=\psig(\epsilon,\Sigma\restr\finA)\restr\rmid,
\end{align*}
where the middle equality follows from~\eqref{eq:active-equal-restr}. This concludes the proof of the lemma.
\end{proof}

\subsection{Properties of \texorpdfstring{$\sigA_P$}{sigP}}
\label{ssec:sig-properties}

\begin{lemma}
The functions $\sigA_P$ defined in~\eqref{eq:def-signature} satisfy Items~\ref{it:sig-inf} to~\ref{it:sig-max} from Lemma~\ref{lem:sig-exist}.
\end{lemma}

\begin{proof}
This proof is a~routine check of the respective cases.

Firstly, Item~\ref{it:sig-inf} follows from the infimum taken in~\eqref{eq:def-signature}, as $\psig(\finA,\Sigma)$ is always a~tuple of ordinals (not $\infty$). Item~\ref{it:sig-neg} follows directly from the fact that if $t=\sneg(t')$ then $\act(\Sigma)=\emptyset$ and $\psig(\epsilon,\Sigma)=(0,\ldots,0)$.

Consider Items~\ref{it:sig-j-win} and~\ref{it:sig-j-los} with a~tree of the form $\spri{\rmid}(t)$ (depending on the case, $\rmid$ is either $P$\=/winning or $P$\=/losing). Let $\finA=0$ that is $t$ is the subtree of $\spri{\rmid}(t)$ under $\finA$. Notice that there is a~bijection $\Sigma\leftrightarrow \Sigma\restr \finA$ between winning strategies of $P$ in $\Gg\big(\spri{\rmid}(t)\big)$ and winning strategies of $P$ in $\Gg(t)$. Thus, it is enough to prove the respective equality for each winning strategy $\Sigma$ separately. Take such a~strategy and consider the two possible cases.

First assume that $\rmid$ is $P$\=/winning. Then $\epsilon\notin\act(\Sigma)$ and $\psig(\epsilon,\Sigma)=\psig(\finA,\Sigma)$---either $\finA\notin\act(\Sigma)$ and both values are defined by the same supremum, or $\finA\in\act(\Sigma)$ and $\psig(\epsilon,\Sigma)$ is the supremum over the singleton set $\suk_{\epsilon}(\Sigma)=\{\finA\}$. By the definition of active nodes, for all $\finB\in\act(\Sigma)$ we know that $t(\finB)<\rmid$. Let $\rmid'=\rmid-1$. By the form of the $P$\=/signatures in Item~\ref{it:inductive-psig-los} of Definition~\ref{def:signatures} we know that $\psig(\epsilon,\Sigma)$ is of the form $(\ordA_{\rmin'},\ldots,\ordA_{\rmid'},0,\ldots,0)$. Applying Lemma~\ref{lem:sig-invariant} to the $P$\=/losing number $\rmid'$ (if $\rmid'<\rmin$ then we are done) we know that
\[\psig(\epsilon,\Sigma\restr\finA)\restr{\rmid'}=\psig(\finA,\Sigma)\restr\rmid'=\psig(\epsilon,\Sigma)\restr\rmid'=
(\ordA_{\rmin'},\ldots,\ordA_{\rmid'}).\]
Therefore, the equality from Item~\ref{it:sig-j-win} of Lemma~\ref{lem:sig-exist} holds for $\Sigma$ and $\Sigma\restr\finA$.

Now assume that $\rmid$ is $P$\=/losing. Then $\epsilon\in\act(\Sigma)$ and by Definition~\ref{def:signatures} we know that $\psig(\epsilon,\Sigma)=(\ordA_{\rmin'},\ldots,\ordA_{\rmid}+1,0,\ldots,0)$ where $\psig(\finA,\Sigma)=(\ordA_{\rmin'},\ldots,\ordA_{\rmax'})$. Again by applying Lemma~\ref{lem:sig-invariant} we get that
\[\psig(\epsilon,\Sigma\restr\finA)\restr\rmid=\psig(\finA,\Sigma)\restr\rmid=(\ordA_{\rmin'},\ldots,\ordA_{\rmax'}),\]
what demonstrates equality in Item~\ref{it:sig-j-los} of Lemma~\ref{lem:sig-exist} for $\Sigma$ and $\Sigma\restr\finA$.

Consider Item~\ref{it:sig-min} with $t=\spla{P}(t_\dL,t_\dR)$. The set $\Ww$ of all winning strategies in $\Gg(t)$ can be split into those who contain the node $(0)$ (i.e.~the left child of the root) denoted $\Ww_{(0)}$ and those who contain the node $(1)$ denoted $\Ww_{(\dR)}$; each category of strategies may be empty. For each strategy $\Sigma\in\Ww_{\finA}$ containing a~child $\finA$ of the root, $\psig(\epsilon,\Sigma)=\psig(\epsilon,\Sigma\restr\finA)$, because $t(\epsilon)=\spla{P}$ and $\act(\Sigma) = \finA\cdot \act(\Sigma\restr\finA)$. Therefore, Item~\ref{it:sig-min} follows from~\eqref{eq:def-signature} as
\[\sigA_P(t)=\inf_{\Sigma\in\Ww}\psig(\epsilon,\Sigma)=\min_{\finA\in\{(0),(1)\}}\ \inf_{\Sigma\in\Ww_{\finA}}\psig(\epsilon,\Sigma\restr\finA)=\min\big\{\sigA_P(t_\dL),\sigA_P(t_\dR)\big\}.\]

Now consider Item~\ref{it:sig-max} with $t=\spla{\bar{P}}(t_\dL,t_\dR)$. There is a~bijection $\Sigma \leftrightarrow\big(\Sigma\restr{(0)},\Sigma\restr{(1)}\big)$ between winning strategies $\Sigma$ in $\Gg(t)$ and pairs of winning strategies $(\Sigma_\dL,\Sigma_\dR)$ in $\Gg(t_\dL)$ and~$\Gg(t_\dR)$ respectively. Therefore, it is enough to prove the equality from Item~\ref{it:sig-max} for each winning strategy $\Sigma$ separately. Take such a~strategy and observe that $\epsilon\notin\act(\Sigma)$ and therefore
\[\suk_\epsilon(\Sigma)=(0)\cdot \suk_\epsilon\big(\Sigma\restr{(0)}\big)\ \cup\ (1)\cdot \suk_\epsilon\big(\Sigma\restr{(1)}\big),\]
Thus, by Item~\ref{it:inductive-psig-rest} of Definition~\ref{def:signatures} and Fact~\ref{ft:restriction} we know that
\begin{align*}
\psig(\epsilon,\Sigma)&=\sup_{\finB\in\suk_\epsilon(\Sigma)}\psig(\finB,\Sigma)\\
&=\sup_{\finB\in\suk_\epsilon(\Sigma)}\psig(\epsilon,\Sigma\restr\finB)\\
&=\max_{\finA\in\{(0),(1)\}}\ \sup_{\finB\in\suk_\epsilon(\Sigma\restr\finA)}\psig\big(\epsilon,\Sigma\restr(\finA\cdot\finB)\big)\\
&=\max_{\finA\in\{(0),(1)\}}\ \sup_{\finB\in\suk_\epsilon(\Sigma\restr\finA)}\psig\big(\finB,\Sigma\restr\finA)\big)\\
&=\max_{\finA\in\{(0),(1)\}}\psig(\epsilon,\Sigma\restr\finA).
\end{align*}
This concludes the proof showing that the constructed functions $\sigA_P$ satisfy Items~\ref{it:sig-inf} to~\ref{it:sig-max} from Lemma~\ref{lem:sig-exist}.
\end{proof}

We can now apply the above observations to demonstrate Lemma~\ref{lem:sig-subopt}.

\newcommand{\lemSigSubopt}{
Let $P\in\{\pI,\pII\}$, $t\in\Wng{P}{\rmin}{\rmax}$, $\rmin'$ be the minimal $P$\=/losing number, and $\rpro$ be some $P$-losing number. Assume that $(\ordA_{\rmin'},\ordA_{\rmin'+2},\ldots,\ordA_{\rpro})$ is a~tuple of ordinals and $\Sigma_P$ is a~partial strategy of the player~$P$ in $\Gg(t)$ such that:
\begin{itemize}
\item $\Sigma_P$ never reaches a~node $\finA$ with $t(\finA)=\spri{\rmid}$ with $\rmid\leq\rpro$ nor a~node $\finA$ with $t(\finA)=\sneg$,
\item infinite plays of $\Sigma_P$ are winning for $P$,
\item if a~position $\finA\in\dom(t)$ is not reachable by $\Sigma$ then $\sigA_P\big(t\restr\finA\big)\restr{\rpro}\lexeq (\ordA_{\rmin'},\ldots,\ordA_{\rpro})$.
\end{itemize}
Under all these assumptions $\sigA_P(t)\restr{\rpro}\lexeq (\ordA_{\rmin'},\ldots,\ordA_{\rpro})$.
}

\begin{lemma}
\label{lem:sig-subopt}
\lemSigSubopt
\end{lemma}

\begin{proof}
Consider a~position $\finA$ that is not reachable by $\Sigma$. By the definition of $\sigA_P(t)$ as an infimum we know that there exists a~winning strategy $\Sigma'_\finA$ of $P$ in $\Gg(t\restr \finA)$ with $\psig(\epsilon,\Sigma'_\finA)=\sigA_P(t\restr \finA)$. By the assumption $\psig(\epsilon,\Sigma'_\finA)\restr\rpro\lexeq (\ordA_{\rmin'},\ldots,\ordA_{\rpro})$. Let $\Sigma'$ be the partial strategy obtained as the extension of $\Sigma$ by the strategies $\Sigma'_\finA$ for all the positions $\finA$ that are not reachable by $\Sigma$:
\begin{equation}
\label{eq:union-of-strat}
\Sigma'\eqdef \Sigma\cup\bigcup_{\text{$\finA$ not reachable by $\Sigma$}}\ \finA\cdot \Sigma'_\finA.
\end{equation}
Clearly $\Sigma'$ is indeed a~strategy and it is winning for $P$. Notice that the set $\act_\epsilon(\Sigma,\rmid)$ is disjoint from $\Sigma$ and therefore we get that
\[\act_\epsilon(\Sigma',\rpro)=\bigcup_{\text{$\finA$ not reachable by $\Sigma$}} \finA\cdot\act_\epsilon(\Sigma'_\finA,\rpro).\]

Thus, by Lemma~\ref{lem:psig-restr} and the assumptions on the strategies $\Sigma'_\finA$ we obtain that:
\begin{align*}
\psig(\epsilon,\Sigma')\restr\rpro&=\sup_{\finB\in \act_\epsilon(\Sigma', \rpro)}\big(\psig(\epsilon,\Sigma'\restr\finB)\restr\rpro\big)\\
&=\sup_{\text{$\finA$ not reachable by $\Sigma$}}\ \sup_{\finB\in \act_\epsilon(\Sigma'_\finA, \rpro)}\big(\psig(\epsilon,\Sigma'_\finA\restr\finB)\restr\rpro\big)\\
&=\sup_{\text{$\finA$ not reachable by $\Sigma$}}\ \big(\psig(\epsilon,\Sigma'_\finA)\restr\rpro\big)\\
&\lexeq(\ordA_{\rmin'},\ldots,\ordA_{\rpro}).
\end{align*}
\end{proof}

\subsection{Consequences of the invariants}
\label{ssec:sig-consequences}

We will now show that any pair of functions $\sigA'_P$ that satisfies the invariants from Lemma~\ref{lem:sig-exist} must guarantee that $\sigA'$\=/optimal strategies are winning, i.e.~Lemma~\ref{lem:sig-opt-win}. For the sake of this subsection assume that $\sigA'_P$ is such a~pair of functions, satisfying Items~\ref{it:sig-inf} to~\ref{it:sig-max} from Lemma~\ref{lem:sig-exist}, $t$ is well\=/formed, $\sigA'_P(t)\lex\infty$, and $\Sigma$ is a~$\sigA'$\=/optimal strategy in $\Gg(t)$ (that is an optimal strategy defined with respect to the functions $\sigA'_P$ for $P=\pI,\pII$).

\begin{lemma}
\label{lem:strat-opt-monotone}
Under the above assumptions, consider a~pair $\finA\preceq\finB\in\Sigma$. Assume that for all $\finB'$ such that $\finA\preceq\finB'\prec\finB$ the value $t(\finB')$ is not $\sneg$ nor $\spri{\rmid'}$ with $\rmid'<\rmid$. Then
\begin{equation}
\label{eq:sig-monotone}
\sigA'_P(t\restr\finA)\restr\rmid\lexgeq \sigA'_P(t\restr\finB)\restr\rmid.
\end{equation}

Moreover, if $t(\finA)=\spri{\rmid}$ then the inequality is strict.
\end{lemma}

\begin{proof}
The thesis follows inductively from the case when $\finB$ is a~child of $\finA$. This case is done using Lemma~\ref{lem:sig-exist} by considering the respective cases of the form of $t\restr \finA$. By our assumptions, the cases from Items~\ref{it:sig-inf} and \ref{it:sig-neg} cannot happen. The cases in Items~\ref{it:sig-j-win} and~\ref{it:sig-j-los} cannot violate~\eqref{eq:sig-monotone} because no priority smaller than $\rmid$ occurs in $t(\finA)$. Moreover, Item~\ref{it:sig-j-los} implies strict inequality in the case $t(\finA)=\spri{\rmid}$. Item~\ref{it:sig-min} together with $\sigA'$\=/optimality of $\Sigma$ guarantee equality in~\eqref{eq:sig-monotone}. Finally, Item~\ref{it:sig-max} implies~\eqref{eq:sig-monotone} without additional assumptions.
\end{proof}

The following result shows that Lemma~\ref{lem:sig-opt-win} is a~consequence of Lemma~\ref{lem:sig-exist}.

\begin{lemma}
\label{lem:sig-strat-witness}
Under the assumptions of this subsection, $\Sigma$ is winning for $P$.
\end{lemma}

\begin{proof}
First we assume contrarily that $\Sigma$ is not winning and $\infA$ is an infinite play of $\Sigma$ that is losing for $P$. Let $\rmid$ be the minimal priority that occurs infinitely many times on $\infA$ in $t$. As $\infA$ is not winning, we know that $\rmid$ is $P$\=/losing.

As $t$ is well\=/formed, from some point on there are no occurrences of $\sneg$ in $t$ on $\infA$. Similarly, from some point on there are no occurrences of priorities below $\rmid$ on $\infA$ in $t$.

Optimal strategies are defined based on the subtrees (without taking into account the history of a~play, except the \emph{switched} or \emph{kept} nodes). Similarly, losing plays are prefix\=/independent. This means that we can restrict our attention to the subtree of $t$ that is sufficiently far on $\infA$ (it might require swapping $P$ with $\bar{P}$ if $\infA$ is \emph{switched}). By this restriction we assume that there is no occurrence of $\sneg$ on $\infA$ in $t$ nor an occurrence of a~priority below $\rmid$ on $\infA$ in $t$.

Now we are ready to obtain a~contradiction, as we assumed that $\rmid$ occurs infinitely many times on $\infA$ in $t$ and no lower priority appears on $\infA$ in $t$. This means that the assumptions of Lemma~\ref{lem:strat-opt-monotone} hold for all $\finA\prec\finB\prec\infA$. Therefore, the sequence of prefixes of the $P$\=/signatures of subtrees of $t$ along $\infA$:
\[\big(\sigA'_P(t_n)\restr \rmid\big)_{n\in\w}\text{, where $t_n=t\restr (\infA\restr{n})$}\]
is a~non\=/increasing sequence that strictly ${\lexeq}$\=/decreases infinitely many times. This violates the fact that the ${\lexeq}$ order on $P$\=/signatures is a~well\=/order.
\end{proof}

\subsection{Optimality of \texorpdfstring{$\sigA_P$}{sigP}}
\label{ssec:sig-optimal}

In this section we show that the pair of function $\sigA_P$ constructed in Subsection~\ref{ssec:sig-construction} is point\=/wise minimal satisfying conditions from Lemma~\ref{lem:sig-exist}. This is not crucial for the rest of the article but shows that the definition is canonical.

\begin{lemma}
\label{lem:sig-opt-strat}
If a~pair of functions $\sigA'_P$ satisfies Items~\ref{it:sig-inf} to~\ref{it:sig-max} from Lemma~\ref{lem:sig-exist}, $t$ is well\=/formed, $P\in\{\pI,\pII\}$, $\sigA'_P(t)\lex\infty$, and $\Sigma$ is a~$\sigA'$\=/optimal strategy in $\Gg(t)$ then $\psig(\epsilon,\Sigma)\lexeq\sigA'_P(t)$.
\end{lemma}

\begin{proof}
By Lemma~\ref{lem:sig-strat-witness} we know that $\Sigma$ is winning. In particular, by Item~\ref{it:sig-inf} of Lemma~\ref{lem:sig-exist} we know that for all $\finA\in\Sigma$ the value $\sigA'_P(t\restr\finA)$ is not $\infty$.

As $\Sigma$ is winning we know by Lemma~\ref{lem:act-is-well-founded} that ${\gg}$ is well\=/founded. We will prove by well\=/founded induction over ${\gg}$ that for a~node $\finA\in\Sigma$ we have
\begin{equation}
\label{eq:inductive-p-to-s}
\psig(\finA,\Sigma) \lexeq \sigA'_P(t\restr\finA).
\end{equation}

The thesis of the lemma follows from this equation with $\finA=\epsilon$. Consider a~node $\finA\in\Sigma$. First consider the case that $\finA\in\act(\Sigma)$ and $t\restr\finA=\spri{\rmid}(t')$ with a~$P$\=/losing number $\rmid$ and $\finA'=\finA\cdot 0$ the unique child of $\finA$ in $t$. Then:
\begin{align*}
\psig(\finA,\Sigma)&=(\ordA_{\rmin'},\ldots,\ordA_{\rmid}{+}1,0,\ldots,0),\\
\psig(\finA',\Sigma)&=(\ordA_{\rmin'},\ldots,\ordA_{\rmid},\ordA_{\rmid'+2},\ldots,\ordA_{\rmax'}),\\
\sigA'_P(t)&=(\ordA'_{\rmin'},\ldots,\ordA'_{\rmid}{+}1,0,\ldots,0),\\
\sigA'_P(t')&=(\ordA'_{\rmin'},\ldots,\ordA'_{\rmid},\ordA'_{\rmid'+2},\ldots,\ordA'_{\rmax'}).
\end{align*}
Thus, the thesis follows from the inductive assumptions about $\finA'$.

Now consider the case that $\finA\notin\act(\Sigma)$. By Item~\ref{it:inductive-psig-rest} of Definition~\ref{def:signatures} it is enough to show that for $\finB\in \suk_\finA(\Sigma)$ we have $\psig(\finB,\Sigma)\lexeq\sigA'_P(t\restr\finA)$.

Consider a~node $\finB\in\suk_\finA(\Sigma)$ with a~$P$\=/losing number $\rmid$ such that $t(\finB)=\spri{\rmid}$. By Item~\ref{it:inductive-psig-los} of Definition~\ref{def:signatures} we know that $\psig(\finA,\Sigma)$ is of the form $(\ordA_{\rmin'},\ldots,\ordA_{\rmid},0,\ldots,0)$, similarly $\sigA'_P(t\restr\finB)$ is of the form $(\ordA'_{\rmin'},\ldots,\ordA'_{\rmid},0,\ldots,0)$. The inductive thesis for $\finB$ says that $\psig(\finB,\Sigma) \lexeq \sigA'_P(t\restr\finB)$.

The pair $\finA\prec \finB$ satisfies the assumptions of Lemma~\ref{lem:strat-opt-monotone}---no priority lower than $\rmid$ nor $\sneg$ occurs between $\finA$ and $\finB$ as otherwise $\finB$ would not belong to $\act(\Sigma)$. The claim implies that $\sigA'_P(t\restr\finB)\restr{\rmid}\lexeq\sigA'_P(t\restr\finA)\restr{\rmid}$ and as $\sigA'_P(t\restr\finB)$ is padded with zeros we also know that $\sigA'_P(t\restr\finB)\lexeq\sigA'_P(t\restr\finA)$. By combining these inequalities we obtain the deserved inequality $\psig(\finB,\Sigma)\lexeq\sigA'_P(t\restr\finA)$.
\end{proof}

\begin{fact}
\label{ft:sig-opt-exist}
If a~pair of functions $\sigA'_P$ satisfies Items~\ref{it:sig-inf} to~\ref{it:sig-max} from Lemma~\ref{lem:sig-exist}, $t$ is well\=/formed, $P\in\{\pI,\pII\}$, and $\sigA'_P(t)\lex\infty$ then there exists a~$\sigA'$\=/optimal strategy $\Sigma$ of $P$ in $\Gg(t)$.
\end{fact}

\begin{proof}
$\Sigma$ is constructed inductively, by following the optimal choices according to Item~\ref{it:sig-min} of Lemma~\ref{lem:sig-exist}. The construction preserves the invariant that if $\finA\in\Sigma$ and $P'=P$ for $\finA$ \emph{kept} and $P'=\bar{P}$ for $\finA$ \emph{switched} then $\sigA'_{P'}(\finA)\lex\infty$.
\end{proof}

\begin{corollary}
If a~pair of functions $\sigA'_P$ for $P\in\{\pI,\pII\}$ satisfies Items~\ref{it:sig-inf} to~\ref{it:sig-max} from Lemma~\ref{lem:sig-exist}, $t$ is well\=/formed, and $P\in\{\pI,\pII\}$ then
\[\sigA_P(t)\lexeq \sigA'_P(t),\]
where $\sigA_P$ is the function defined in Subsection~\ref{ssec:sig-construction}.
\end{corollary}

\begin{proof}
Consider a~tree $t\in\trees_A$ and a~player $P\in\{\pI,\pII\}$. If $P$ loses $\Gg(t)$ then both values are equal $\infty$. Assume that $P$ wins $\Gg(t)$, in that case $\sigA'_P(t)\lex\infty$. Let $\Sigma$ be a~$\sigA'$\=/optimal strategy from Fact~\ref{ft:sig-opt-exist}. Lemma~\ref{lem:sig-opt-strat} implies that $\psig(\epsilon,\Sigma)\lexeq\sigA'_P(t)$. Since $\sigA_P(t)\lexeq \psig(\epsilon,\Sigma)$ because of the infimum in Definition~\ref{def:signatures}, the thesis follows.
\end{proof}

\section{Construction of \texorpdfstring{$f$}{f}}
\label{sec:hardness}

Fix a~pair $\rmin <\rmax$. In this section we will prove the following proposition.

\begin{proposition}
\label{pro:W-reduction}
There exists a~continuous function $\fun{f}{\trees_{\Ang{\rmin}{\rmax}}}{\trees_{\Angp{\rmin}{\rmax}}}$ such that:
\begin{itemize}
\item If $t$ is well\=/formed then $f(t)$ is also well\=/formed.
\item For a~well\=/formed $t$ and a~player $P$ we have: $t\in\Wng{P}{\rmin}{\rmax}$ if and only if $f(t)\in\lan{P}{\rmin}{\rmax}$.
\end{itemize}
\end{proposition}

Before proving that proposition, notice that a~tree over the alphabet $\Alp{\rmin}{\rmax}$ can be seen as a~well\=/formed tree over $\Ang{\rmin}{\rmax}$ and $\Wng{\pI}{\rmin}{\rmax}\cap\trees_{\Alp{\rmin}{\rmax}}=\Win{\pI}{\rmin}{\rmax}$ which is topologically equivalent to $\W{\rmin}{\rmax}$. Therefore, the above proposition implies that $f\restr{\trees_{\Alp{\rmin}{\rmax}}}$ satisfies the assumptions of Corollary~\ref{cor:reduction}. As $\rmin<\rmax$ are arbitrary, Theorem~\ref{thm:main} follows.

The following lemma claims the combinatorial core of this article: it shows that one can compare the $P$\=/signatures using a~continuous reduction to the languages $\Wng{P}{\rmin}{\rmax}$.

\newcommand{\lemReduction}{
There exists a~continuous function $\fun{c_P}{\big(\trees_{\Ang{\rmin}{\rmax}}\big)^2}{\trees_{\Ang{\rmin}{\rmax}}}$ such that if $t_\dL$ and $t_\dR$ are well\=/formed then so is $c_P(t_\dL,t_\dR)$ and additionally
\[c_P(t_\dL,t_\dR)\in\Wng{\pI}{\rmin}{\rmax}\ \text{ if and only if }\ \sigA_P(t_\dL)\lexeq \sigA_P(t_\dR).\]
}

\begin{lemma}
\label{lem:reduction}
\lemReduction
\end{lemma}
%
%
The rest of this section demonstrates Proposition~\ref{pro:W-reduction}. Lemma~\ref{lem:reduction} is proved in the next section.

Consider a~function $\fun{f}{\trees_{\Ang{\rmin}{\rmax}}}{\trees_{\Angp{\rmin}{\rmax}}}$ defined recursively as:
\begin{align*}
f\big(\spla{P}(t_\dL,t_\dR)\big)&=\splp{P}\Big(f(t_\dL),f(t_\dR),f\big(c_P(t_\dL,t_\dR)\big)\Big)&\text{for $P\in\{\pI,\pII\}$,}\\
f\big(\sneg(t)\big)&=\sneg\big(f(t)\big),\\
f\big(\spri{\rmid}(t)\big)&=\spri{\rmid}\big(f(t)\big)&~\text{for $\rmid\in\{\rmin,\ldots,\rmax\}$.}
\end{align*}
Clearly by the definition of $f$ we know that $\shave\big(f(t)\big)=t$. Additionally, $f(t)$ is defined recursively using $c_P$ which is continuous, therefore $f$ is also continuous. As $c_P$ maps well\=/formed trees to well\=/formed trees, so does $f$.

First assume that $f(t)\in\lan{P}{\rmin}{\rmax}$ as witnessed by an accepting run $\rho$ of $\Uu$ over $f(t)$ with $\rho(\epsilon)=(\star,P,\star)$. Fact~\ref{ft:full-to-win} says that $P$ wins $\Gg\big(\shave(f(t))\big)=\Gg(t)$, so $t\in\Wng{P}{\rmin}{\rmax}$.

For the converse assume that $P_0$ wins $\Gg(t)$ for a~well\=/formed tree $t$ over $\Ang{\rmin}{\rmax}$.

\newcommand{\lemConstructRun}[1]{
If $t$ is well\=/formed then there exists a~unique run $\rho$ of $\Uu$ over $f(t)$ such that for every $\finA\in\dom\big(f(t)\big)$ we have
\begin{equation}
#1{eq:run-invariant}
\rho(\finA)=(\star,P,\star)\quad\text{ if and only if }\quad\text{$P$ wins $\Gg\big(\shave(f(t)\restr \finA)\big)$}.
\end{equation}
Moreover, all the $\rho$\=/strategies are winning for the respective players.
}

\begin{lemma}
\label{lem:construct-run}
\lemConstructRun{\label}
\end{lemma}

\begin{proof}
The construction of $\rho$ is inductive from the root preserving~\eqref{eq:run-invariant}. The only ambiguity when choosing transitions of $\Uu$ is when we reach a~node $\finB\in\dom\big(f(t)\big)$ such that $\rho(\finB)$ is of the form $(\star,P,\star)$ and $f(t)\restr \finB=\splp{P}\big(f(t_\dL),f(t_\dR),f(c_P(t_\dL,t_\dR))\big)$. We choose either the left or the right of the two lower transitions of $\Uu$ depending on the winner in $\Gg\big(\shave(f(t)\restr\finB\cdot 2)\big)$ in such a~way to satisfy~\eqref{eq:run-invariant} for $\finA=\finB\cdot 2$. By the symmetry assume that we used the left transition. This leaves undeclared the second coordinate of $\rho(\finB\cdot 1)$ (resp. $\rho(\finB\cdot 0)$ in the case of the right transition). Again we declare this coordinate accordingly to~\eqref{eq:run-invariant}. We need to check that~\eqref{eq:run-invariant} is also satisfied for $\finA=\finB\cdot 0$ (resp. $\finA=\finB\cdot 1$) i.e.~that $P$ wins $\Gg\big(\shave(f(t)\restr\finB\cdot 0)\big)$. To see that, we notice that the following conditions are equivalent $(\bigstar)$:
\newcommand{\lann}[1]{\hfill$[$ #1 $]$}
\begin{itemize}
\item $\rho(\finB)=(\star,\star,\dL)$ \lann{by the form of the transitions of $\Uu$}
\item $\rho(\finB\cdot 2) = (\star, \pI,\star)$ \lann{by the definition of $\rho$}
\item $\pI$ wins in $\Gg\big(\shave(f(t)\restr (\finB\cdot 2))\big)$ \lann{by the form of $f(t)\restr(\finB\cdot 2)$}
\item $\pI$ wins in $\Gg\big(\shave(f(c_P(t_\dL,t_\dR))\big)$ \lann{by the equality $\shave\big(f(t')\big)=t'$}
\item $\pI$ wins in $\Gg\big(c_P(t_\dL,t_\dR)\big)$ \lann{by Lemma~\ref{lem:reduction}}
\item $\sigA_P(t_\dL)\lexeq \sigA_P(t_\dR)$.
\end{itemize}

Thus, if we choose the lower left transition of $\Uu$, we know that $\sigA_P(t_\dL)\lexeq \sigA_P(t_\dR)$. By the inductive invariant we know that $\infty\lexge\sigA_P\big(\shave(f(t)\restr\finB)\big)=\sigA_P\big(\spla{P}(t_\dL,t_\dR)\big)$. Therefore, Item~\ref{it:sig-min} of Lemma~\ref{lem:sig-exist} implies that $\sigA_P(t_\dL)\lex\infty$ so in fact $P$ wins $\Gg(t_\dL)=\Gg\big(\shave(f(t_\dL))\big)=\Gg\big(\shave(f(t)\restr\finB\cdot 0)\big)$. Thus, the invariant~\eqref{eq:run-invariant} is also preserved for $\finA=\finB\cdot 0$. This concludes the inductive definition of $\rho$. Lemma~\ref{lem:second-to-other} implies uniqueness.

Take $\finA\in\dom\big(f(t)\big)$ with $\rho(\finA)=(\star,P,\star)$. Let $\Sigma$ be the $\rho$\=/strategy in $\finA$ and $r'=f(t)\restr\finA$. Consider a~node $\finB\in\dom\big(\shave(r')\big)$ such that $\shave(r')(\finB)=\spla{P'}$ with $P'=P$ if $\finB$ is \emph{kept} and $P'=\bar{P}$ if $\finB$ is \emph{switched}. In both cases $f(t)(\finA\cdot\finB)=\splp{P'}$. By the above equivalence $(\bigstar)$ and the definition of a~$\rho$\=/strategy, $\Sigma$ makes a~$\sigA$\=/optimal move in $\finB$. Therefore, $\Sigma$ is optimal. Invariant~\eqref{eq:run-invariant} says that $P$ wins $\Gg\big(\shave(r')\big)$, so Lemma~\ref{lem:sig-opt-win} implies that $\Sigma$ is winning.
\end{proof}

Fix the run $\rho$ given by the above lemma. Since $\Gg(t)=\Gg\big(\shave(f(t))\big)$, $\rho(\epsilon)=(\star,P_0,\star)$. Since all the $\rho$\=/strategies are winning, $\rho$ is accepting by Lemma~\ref{lem:correct-strategy} and $f(t)\in\lan{P}{\rmin}{\rmax}$. This concludes the proof of Theorem~\ref{thm:main} assuming that Lemma~\ref{lem:reduction} holds.

\section{Comparing signatures}
\label{sec:comparing}

We will now prove Lemma~\ref{lem:reduction} by defining, given two trees $\tl$ and $\tr$ over $\Ang{\rmin}{\rmax}$, a~game $\Cc_P(\tl,\tr)$. To indicate the difference between the games $\Cc_P$ and $\Gg$, we denote the players of $\Cc_P$ as $\eve= \pI$ and $\adam=\pII$. The purpose of the game $\Cc_P$ will be to ensure that $\eve$ wins $\Cc_P(\tl,\tr)$ if and only if $\sigA_P(\tl)\lexeq\sigA_P(\tr)$. The winning condition of the game $\Cc_P$ will be a~parity condition, however, the game will allow certain \emph{lookahead} (see steps \EL and \AL in the definition of $\Cc_P$). Then, the function $c_P$ will just unravel the game $\Cc_P(\tl,\tr)$ into a~tree over the alphabet $\Ang{\rmin}{\rmax}$.

The construction is an analogue of the rank method, as discussed in Section~\ref{ap:ranks}.

If $\rmin'$ and $\rmax'$ are the minimal and maximal $P$\=/losing numbers; $\rmid$ is a~$P$\=/losing number; and $\sigA=(\ordA_{\rmin'},\ordA_{\rmin'+2},\ldots,\ordA_{\rmax'})$ is a~$P$\=/signature then $\sigA\restr{\rmid}$ is the tuple $(\ordA_{\rmin'},\ordA_{\rmin'+2},\ldots,\ordA_{\rmid})$. For completeness let $\infty\restr{\rmid}\eqdef\infty$. Clearly $\sigA\restr{\rmax'}=\sigA$. Notice that $\sigA\restr{\rmid}$ is also a~$P$\=/signature (with $\rmax=\rmid$) and moreover if $\sigA\lexeq\sigA'$ then $\sigA\restr{\rmid}\lexeq\sigA'\restr{\rmid}$.

A~position of the game $\Cc_P$ is a~triple $\big(\tl,\tr,\rpro\big)$ where $\tl,\tr\in\trees_{\Ang{\rmin}{\rmax}}$ and $\rpro$ is a~$P$\=/losing number. As $\Cc_P(\tl,\tr)$ we denote the game $\Cc_P$ with the initial position set to $\big(\tl,\tr,\rmax'\big)$. The game is designed in such a~way to guarantee the following claim.

\begin{claim}
\label{cl:winner-in-C}
A~position $\big(\tl,\tr,\rpro\big)$ is winning for \eve in $\Cc_P$ if and only if
\begin{equation}
\label{eq:winner-in-C}
\sigA_P(\tl)\restr{\rpro}\ \lexeq\ \sigA_P(\tr)\restr{\rpro}.
\end{equation}
\end{claim}

A~single round of the game $\Cc_P$ consists of a~sequence of choices done by the players. It is easy to encode such a~sequence using additional intermediate positions of the game. For the sake of readability, we do not specify these positions explicitly. Instead, a~round moves the game from a~position $\big(\tl,\tr,\rpro\big)$ into a~new position according to the following sequential steps:
\begin{description}
\item[\EL] \eve can claim that $P$ loses $\Gg(\tr)$. If she does so, the game ends and \eve wins iff $\tr\notin\Wng{P}{\rmin}{\rmax}$.
\item[\AL] \adam can claim that $P$ loses $\Gg(\tl)$. If he does so, the game ends and \adam wins iff $\tl\notin\Wng{P}{\rmin}{\rmax}$.
\item[\EI] \eve can modify $\rpro$ into another $P$\=/losing number $\rpro'< \rpro$. In that case the round ends and the next position is $\big(\tl',\tr,\rpro'\big)$ where $\tl'=\spri{\rpro'}(\tl)$.
\item[\AI] \adam can modify $\rpro$ into another $P$\=/losing number $\rpro'<\rpro$. In that case the round ends and the next position is $\big(\tl,\tr,\rpro'\big)$.
\item[\Db] If $\tl=\spri{\rpro}(\tl')$ and $\tr=\spri{\rpro}(\tr')$ then the round ends and the next position is $(\tl',\tr',\rpro)$.
\item[\Dl] If $\tl$ is not of the form $\spri{\rpro}(\tl'')$ then a~step called $\downl$ is done, resulting in an~immediate win of \eve or a~new tree $\tl'$. The round ends and the next position is $(\tl',\tr,\rpro)$.
\item[\Dr] Otherwise $\tl$ is of the form $\spri{\rpro}(\tl')$ and a~step called $\downr$ is done, resulting in an~immediate win of \adam or a~new tree~$\tr'$. The round ends and the next position is $(\tl,\tr',\rpro)$.
\end{description}
The result $\tl'$ of $\downl$ depends on the form of $\tl$ as follows:
\begin{itemize}
\item If $\tl=\spla{P}(\tl_\dL,\tl_\dR)$ then \eve chooses to set $\tl'=\tl_\dL$ or $\tl'=\tl_\dR$.
\item If $\tl=\spla{\bar{P}}(\tl_\dL,\tl_\dR)$ then \adam chooses to set $\tl'=\tl_\dL$ or $\tl'=\tl_\dR$.
\item If $\tl=\spri{\rmid}(\tl')$ and $\rmid > \rpro$ then $\tl'$ is defined and that round of $\Cc_P$ has priority $\rmid{+}1{-}P$.
\item Otherwise \eve immediately wins.
\end{itemize}
Dually, the result $\tr'$ of $\downr$ depends on the form of $\tr$ as follows:
\begin{itemize}
\item If $\tr=\spla{P}(\tr_\dL,\tr_\dR)$ then \adam chooses to set $\tr'=\tr_\dL$ or $\tr'=\tr_\dR$.
\item If $\tr=\spla{\bar{P}}(\tr_\dL,\tr_\dR)$ then \eve chooses to set $\tr'=\tr_\dL$ or $\tr'=\tr_\dR$.
\item If $\tr=\spri{\rmid}(\tr')$ with $\rmid>\rpro$ then $\tr'$ is defined and that round of $\Cc_P$ has priority $\rmid{-}2{+}P$.
\item Otherwise \adam immediately wins.
\end{itemize}
The rounds of $\Cc_P$ which priority is not declared above have priority $\rmax$. An infinite play $\plyA$ of~$\Cc_P$ is won by \eve if the least priority seen infinitely often during $\plyA$ is even.

\newcommand{\lemUnravelling}{
The game $\Cc_P(\tl,\tr)$ can be unravelled as a~tree $c_P(\tl,\tr)$ over the alphabet $\Ang{\rmin}{\rmax}$ in such a~way that for well\=/formed trees $\tl$, $\tr$, the tree $c_P(\tl,\tr)$ is well\=/formed and \eve wins $\Cc_P(\tl,\tr)$ if and only if $\pI$ wins $\Gg\big(c_P(\tl,\tr)\big)$. Moreover, the function $c_P$ is continuous.
}

\begin{lemma}
\label{lem:unravelling}
\lemUnravelling
\end{lemma}


\begin{proof}
Let $\rmin'$ and $\rmax'$ be the minimal and the maximal $P$\=/losing numbers. We will define recursively a~function $c_P\colon(\tl,\tr,\rpro)\mapsto t\in\trees_{\Ang{\rmin}{\rmax}}$ unravelling the game $\Cc_P$ from a~position~$(\tl,\tr,\rpro)$ into a~tree $t$. The function $c_P(\tl,\tr)$ will be defined as $c_P(\tl,\tr,\rmax')$.

Consider a~$P$\=/losing number $\rpro$. Let $\rmin'=\rpro_0<\rpro_1<\ldots<\rpro_{n-1}<\rpro$ be the list of all $P$\=/losing numbers smaller than $\rpro$. If $P=\pI$ let $w_P=\spri{\rmax}$ and $w_{\sim P}=\sneg$; if $P=\pII$ the symbols are swapped: $w_P=\sneg$ and $w_{\sim P}=\spri{\rmax}$. Let $t_\eve$ be any tree in $\Wng{\eve}{\rmin}{\rmax}$ and $t_\adam$ be any tree in $\Wng{\adam}{\rmin}{\rmax}$. In the definition of $c_P$ we use the symbols \eve and \adam to name the players to follow the convention of $\Cc_P$, however as $\eve=\pI$ and $\adam=\pII$, this obeys the format of the alphabet $\Ang{\rmin}{\rmax}$.

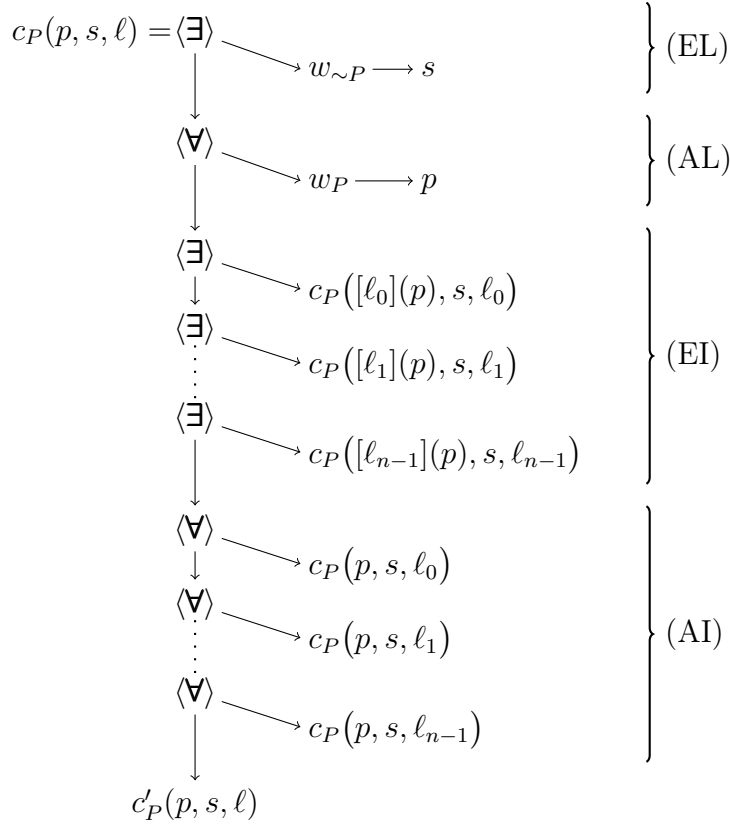
\begin{figure}
\centering
\begin{tikzpicture}
\newcommand{\hsc}{1.5}
\newcommand{\vsc}{-0.5}
\node[toL] at (-0.3, -0) {$c_P(\tl,\tr,\rpro)=$};
\node[letter] (c0) at ( 0.0, -0) {$\spla{\eve}$};
\node[letter, toR] (s0) at ($(c0)+( \hsc, \vsc)$) {$w_{\sim P}$};
\node[letter, toR] (u0) at ($(c0)+(\hsc+1.5, \vsc)$) {$\tr$};

\node[letter] (c1) at ($(c0)+( 0.0, -1.5)$) {$\spla{\adam}$};
\node[letter, toR] (s1) at ($(c1)+( \hsc, \vsc)$) {$w_{P}$};
\node[letter, toR] (u1) at ($(c1)+(\hsc+1.5, \vsc)$) {$\tl$};

\node[letter] (c2) at ($(c1)+( 0.0, -1.5)$) {$\spla{\eve}$};
\node[letter, toR] (s2) at ($(c2)+( \hsc, \vsc)$) {$c_P\big(\spri{\rpro_0}(\tl),\tr,\rpro_0\big)$};

\node[letter] (c3) at ($(c2)+( 0.0, -1.0)$) {$\spla{\eve}$};
\node[letter, toR] (s3) at ($(c3)+( \hsc, \vsc)$) {$c_P\big(\spri{\rpro_1}(\tl),\tr,\rpro_1\big)$};

\node[letter] (c4) at ($(c3)+( 0.0, -1.2)$) {$\spla{\eve}$};
\node[letter, toR] (s4) at ($(c4)+( \hsc, \vsc)$) {$c_P\big(\spri{\rpro_{n-1}}(\tl),\tr,\rpro_{n-1}\big)$};

\node[letter] (c5) at ($(c4)+( 0.0, -1.5)$) {$\spla{\adam}$};
\node[letter, toR] (s5) at ($(c5)+( \hsc, \vsc)$) {$c_P\big(\tl,\tr,\rpro_0\big)$};

\node[letter] (c6) at ($(c5)+( 0.0, -1.0)$) {$\spla{\adam}$};
\node[letter, toR] (s6) at ($(c6)+( \hsc, \vsc)$) {$c_P\big(\tl,\tr,\rpro_1\big)$};

\node[letter] (c7) at ($(c6)+( 0.0, -1.2)$) {$\spla{\adam}$};
\node[letter, toR] (s7) at ($(c7)+( \hsc, \vsc)$) {$c_P\big(\tl,\tr,\rpro_{n-1}\big)$};

\node[letter] (c8) at ($(c7)+( 0.0, -1.5)$) {$c'_P(\tl,\tr,\rpro)$};

\draw (6,-0.8) edge[rbrace] node{\EL} ++(0, 1.2);
\draw (6,-2.3) edge[rbrace] node{\AL} ++(0, 1.2);
\draw (6,-6.0) edge[rbrace] node{\EI} ++(0, 3.4);
\draw (6,-9.7) edge[rbrace] node{\AI} ++(0, 3.4);

\draw[transs] (c0) -- (s0.west);
\draw[transs] (s0.mid east) -- (u0.mid west);
\draw[transs] (c0) -- (c1);

\draw[transs] (c1) -- (s1.west);
\draw[transs] (s1.mid east) -- (u1.mid west);
\draw[transs] (c1) -- (c2);

\draw[transs] (c2) -- (s2.west);
\draw[transs] (c2) -- (c3);

\draw[transs] (c3) -- (s3.west);
\draw[trdots] (c3) -- (c4);

\draw[transs] (c4) -- (s4.west);
\draw[transs] (c4) -- (c5);

\draw[transs] (c5) -- (s5.west);
\draw[transs] (c5) -- (c6);

\draw[transs] (c6) -- (s6.west);
\draw[trdots] (c6) -- (c7);

\draw[transs] (c7) -- (s7.west);
\draw[transs] (c7) -- (c8);
\end{tikzpicture}
\caption{A~part of the tree $c_P(\tl,\tr,\rpro)$ simulating the rounds \EL, \AL, \EI, and \AI of $\Cc_P$. The tree $c'_P(\tl,\tr,\rpro)$ is defined separately, depending on the shape of $\tl$ and $\tr$. Formally, to adhere to the format, the above tree should be padded everywhere with priorities $\rmax$, so that every second node is a~priority node.}
\label{fig:def-of-c}
\end{figure}

A~recursive formula for the tree $c_P(\tl,\tr,\rpro)$ is given in Figure~\ref{fig:def-of-c}, where the subtree $c'_P(\tl,\tr,\rpro)$ depends on the form of the trees $\tl$ and $\tr$ as follows:
\begin{itemize}
\item The case of \Db:
\begin{align*}
c'_P\big(\spri{\rpro}(\tl),\ \spri{\rpro}(\tr),\ \rpro\big)&=\spri{\rmax}\big(c_P(\tl,\tr,\rpro)\big)
\end{align*}
\item The case of \downl:
\begin{align*}
c'_P\big(\spla{P}(\tl_\dL,\tl_\dR),\ \tr,\ \rpro\big)&=\spla{\eve}\big(c_P(\tl_\dL,\tr,\rpro),\ c_P(\tl_\dR,\tr,\rpro)\big)\\
c'_P\big(\spla{\bar{P}}(\tl_\dL,\tl_\dR),\ \tr,\ \rpro\big)&=\spla{\adam}\big(c_P(\tl_\dL,\tr,\rpro),\ c_P(\tl_\dR,\tr,\rpro)\big)\\
c'_P\big(\spri{\rmid}(\tl),\ \tr,\ \rpro\big)&=\spri{\rmid{+}1{-}P}\big(c_P(\tl,\tr,\rpro)\big)&\text{if $\rmid>\rpro$}\\
c'_P\big(\spri{\rmid}(\tl),\ \tr,\ \rpro\big)&=t_\eve&\text{if $\rmid<\rpro$}\\
c'_P\big(\sneg(\tl),\ \tr,\ \rpro\big)&=t_\eve
\end{align*}
\item The case of \downr:
\begin{align*}
c'_P\big(\spri{\rpro}(\tl),\ \spla{P}(\tr_\dL,\tr_\dR),\ \rpro\big)&=\spla{\adam}\big(c_P(\tl,\tr_\dL,\rpro),\ c_P(\tl,\tr_\dR,\rpro)\big)\\
c'_P\big(\spri{\rpro}(\tl),\ \spla{\bar{P}}(\tr_\dL,\tr_\dR),\ \rpro\big)&=\spla{\eve}\big(c_P(\tl,\tr_\dL,\rpro),\ c_P(\tl,\tr_\dR,\rpro)\big)\\
c'_P\big(\spri{\rpro}(\tl),\ \spri{\rmid}(\tr),\ \rpro\big)&=\spri{\rmid{-}2{+}P}\big(c_P(\tl,\tr,\rpro)\big)&\text{if $\rmid>\rpro$}\\
c'_P\big(\spri{\rpro}(\tl),\ \spri{\rmid}(\tr),\ \rpro\big)&=t_\adam&\text{if $\rmid<\rpro$}\\
c'_P\big(\spri{\rpro}(\tl),\ \sneg(\tr),\ \rpro\big)&=t_\adam
\end{align*}
\end{itemize}

The possibility of an immediate win of the players $\eve$ and $\adam$ in $\Cc_P$ is simulated by putting respectively subtrees $t_\eve$ and $t_\adam$ in $c'_P(\tl,\tr,\rpro)$

Notice that in both cases when a~round of $\Cc_P$ has an explicitly declared priority, that priority is $\rmid$ or $\rmid{-}1$ with $\rmin\leq\rpro<\rmid\leq\rmax$. Therefore, the priorities of $\Cc_P$ are within $\{\rmin,\ldots,\rmax\}$ and thus, the tree produced by $c_P$ uses only symbols from the alphabet $\Ang{\rmin}{\rmax}$.

The recursive definition of $c_P$ guarantees that the function is continuous (between every two recursive calls a~new symbol of the output is produced). Moreover, the only place where the function produces the symbol $\sneg$ is in the simulation of the steps \EL and \AL, see Figure~\ref{fig:def-of-c}. But no branch of $c_P(\tl,\tr,\rpro)$ passes through more than one such newly produced symbol $\sneg$. Therefore, if the trees $\tl$ and $\tr$ are well\=/formed then so is $c_P(\tl,\tr,\rpro)$.

The game $\Gg\big(c_P(\tl,\tr,\rpro)\big)$ faithfully represents the choices of the players in $\Cc_P$, therefore \eve wins $\Cc_P$ from a~position $(\tl,\tr,\rpro)$ if and only if $c_P(\tl,\tr,\rpro)\in\Wng{\pI}{\rmin}{\rmax}$.
\end{proof}

Notice that the rules of the game $\Cc_P$ do not allow to move from a~position with a~tree (either $\tl$ or $\tr$) of the form $\sneg(t)$ to a~position with the respective tree being $t$. Therefore, we never need to swap the considered player $P$ into $\bar{P}$.

Claim~\ref{cl:winner-in-C} together with Lemma~\ref{lem:unravelling} prove Lemma~\ref{lem:reduction}. Thus, the rest of this section is devoted to a~proof of Claim~\ref{cl:winner-in-C}. Since the winning condition of $\Cc_P$ is a~parity condition, that game is positionally determined. Thus, to prove Claim~\ref{cl:winner-in-C} it is enough to show that none of the following two cases is possible for a~position $(\tl,\tr,\rpro)$ of $\Cc_P$: 
\begin{itemize}
\item \eqref{eq:winner-in-C} is true and \adam has a~positional winning strategy $\Sigma_\adam$ from $(\tl,\tr,\rpro)$,
\item \eqref{eq:winner-in-C} is false and \eve has a~positional winning strategy $\Sigma_\eve$ from $(\tl,\tr,\rpro)$. 
\end{itemize}



In both cases we will confront the assumed strategy with a~specially designed positional quasi-strategy of the opponent ($\quas{\eve}$ and $\quas{\adam}$ respectively). The quasi-strategy $\quas{\eve}$ will be defined only in positions that satisfy~\eqref{eq:winner-in-C} and the quasi-strategy $\quas{\adam}$ in the remaining positions.

The quasi-strategy $\quas{\eve}$ (resp. $\quas{\adam}$) of a~player \eve (resp. \adam) in a~round starting in a~position~$\big(\tl,\tr,\rpro\big)$ performs the following choices in sub-rounds \EL to \AI:
\begin{itemize}
\item In \EL $\quas{\eve}$ claims that $P$ loses $\Gg(\tr)$ if and only if he really does.
\item In \AL $\quas{\adam}$ claims that $P$ loses $\Gg(\tl)$ if and only if he really does.
\item In \EI $\quas{\eve}$ modifies $\rpro$ into $\rpro'$ if $\rpro'<\rpro$ is the minimal $P$\=/losing number such that $\sigA_P(\tl)\restr{\rpro'}\lex\sigA_P(\tr)\restr{\rpro'}$. If there is no such number, $\quas{\eve}$ does not declare $\rpro'$.
\item In \AI $\quas{\adam}$ modifies $\rpro$ into $\rpro'$ if $\rpro'<\rpro$ is the minimal $P$\=/losing number such that $\sigA_P(\tl)\restr{\rpro'}\lexge\sigA_P(\tr)\restr{\rpro'}$. If there is no such number, $\quas{\adam}$ does not declare $\rpro'$.
\end{itemize}

Moreover, in $\downl$ when $\tl=\spla{P}(\tl_\dL,\tl_\dR)$ the quasi\=/strategy $\quas{\eve}$ chooses to set $\tl'=\tl_\dL$ if and only if $\sigA_P(\tl_\dL)\lexeq\sigA_P(\tl_\dR)$. Dually, in $\downr$ when $\tr=\spla{P}(\tr_\dL,\tr_\dR)$ the quasi\=/strategy~$\quas{\adam}$ chooses to set $\tr'=\tr_\dL$ if and only if $\sigA_P(\tr_\dL)\lexeq\sigA_P(\tr_\dR)$.

Now it remains to define the choices of the quasi-strategies in the steps $\downr$ and $\downl$ when the given tree is of the form $\spla{\bar{P}}(t_\dL,t_\dR)$. This is the place where the choices of $\quas{\eve}$ and $\quas{\adam}$ are not unique and that is why these are quasi\=/strategies.

\begin{definition}
\label{def:preservation}
The quasi\=/strategies $\quas{\eve}$ and $\quas{\adam}$ need to satisfy the following \emph{preservation guarantees}. First, in $\downr$ when $\tr=\spla{\bar{P}}(\tr_\dL,\tr_\dR)$ then $\quas{\eve}$ can set $\tr'$ as any of the two $\tr_\dL$, $\tr_\dR$ that satisfies $\sigA_P(\tr')\restr{\rpro}\lexgeq \sigA_P(\tl)\restr{\rpro}$. Second, in $\downl$ when $\tl=\spla{\bar{P}}(\tl_\dL,\tl_\dR)$ then $\quas{\adam}$ can set $\tl'$ as any of the two $\tl_\dL$, $\tl_\dR$ that satisfies $\sigA_P(\tl')\restr{\rpro}\lexge \sigA_P(\tr)\restr{\rpro}$.
\end{definition}

\subsection{Local properties of the \texorpdfstring{quasi\=/strategies}{quasi-strategies}}

In this subsection we will prove the following two \emph{local} properties of the quasi\=/strategies $\quas{\eve}$ and $\quas{\adam}$.

\newcommand{\ftLeavesChoice}{
In both cases the preservation guarantee leaves at least one possible choice.
}

\begin{fact}
\label{ft:leaves-choice}
\ftLeavesChoice
\end{fact}


\newcommand{\lemHonestSafety}{
Consider a~position $\big(\tl,\tr,\rpro\big)$. If it satisfies~\eqref{eq:winner-in-C} and \eve follows her quasi\=/strategy~$\quas{\eve}$ then either she immediately wins or the next position also satisfies~\eqref{eq:winner-in-C}.

Dually, if the position violates~\eqref{eq:winner-in-C} and \adam follows his quasi\=/strategy $\quas{\adam}$ then either he immediately wins or the next position also violates~\eqref{eq:winner-in-C}.
}

\begin{lemma}
\label{lem:honest-safety}
\lemHonestSafety
\end{lemma}

\paragraph*{The case of \texorpdfstring{\eve}{E}}

First consider the case of $\eve$, her quasi\=/strategy $\quas{\eve}$, and a~position $(\tl,\tr,\rpro)$ that satisfies~\eqref{eq:winner-in-C}, what means that $\sigA_P(\tl)\restr{\rpro}\lexeq\sigA_P(\tr)\restr{\rpro}$. We will consider the successive round of $\Cc_P$ in which $\eve$ plays according to $\quas{\eve}$ and prove the above two statements simultaneously.

First consider the step $\EL$ in which \eve may claim that $P$ loses $\Gg(\tr)$. If she does so then she wins the game, otherwise $P$ wins $\Gg(\tr)$ what means that $\sigA_P(\tr)\lex\infty$ and therefore also $\sigA_P(\tl)\lex\infty$ and $P$ wins $\Gg(\tl)$.

Now, as $P$ wins $\Gg(\tl)$ if \adam decides to make a~declaration in \AL then he loses. Otherwise the game moves to the successive step.

Let $\sigA_P(\tl)=(\ordA_{\rmin'},\ldots,\ordA_{\rmax'})$. If \eve declares a~new $P$\=/losing number $\rpro'\lex\rpro$ then she knows that $\sigA_P(\tl)\restr\rpro'\lex \sigA_P(\tr)\restr\rpro'$. Thus, $(\ordA_{\rmin'},\ldots,\ordA_{\rpro'})\lex \sigA_P(\tr)\restr\rpro'$. This means that 
\[\sigA_P(\tl')\restr{\rpro'}=\sigA_P\big(\spri{\rpro'}(\tl)\big)\restr\rpro'=(\ordA_{\rmin'},\ldots,\ordA_{\rpro'}+1)\lexeq\sigA_P(\tr)\restr\rpro',\]
and in that case the invariant~\eqref{eq:winner-in-C} is preserved.

Otherwise, we know that \eve did not declare such a~number $\rpro'$ and $\sigA_P(\tr)\restr \rpro$ is of the form $(\ordA_{\rmin'},\ldots,\ordA_{\rpro-2},\ordA'_{\rpro})$ where all the coordinates agree with $\sigA_P(\tl)\restr\rpro$ except the last. By the assumption of~\eqref{eq:winner-in-C} we know that $\ordA_{\rpro}\leq\ordA'_{\rpro}$.

If \adam decides to declare a~new number $\rpro'<\rpro$ in \AI then clearly the invariant (with equality) is preserved.

Now consider the three cases \Db, \Dl, and \Dr. First, in the case \Db we know that both numbers $\ordA_{\rpro}$ and $\ordA'_{\rpro}$ are successor ordinals and when moving to $\tl'$ and $\tr'$ they are decreased exactly by one. Therefore, the invariant is still preserved in that case.

Consider the case of \Dl. If $\tl(\epsilon)=\spla{P}$ and \eve plays in $\downl$ with her quasi\=/strategy, she moves to the subtree of lower value $\sigA_P(\tl')$. By Item~\ref{it:sig-min} of Lemma~\ref{lem:sig-exist} we know that in that case $\sigA_P(\tl')=\sigA_P(\tl)$ and the invariant is preserved. If $\tl(\epsilon)=\spla{\bar{P}}$ then by Item~\ref{it:sig-max} of Lemma~\ref{lem:sig-exist}, no matter which subtree \adam chooses, we know that $\sigA_P(\tr')\lexgeq \sigA_P(\tr)$ and therefore the invariant is also preserved. If $\tl(\epsilon)=\spri{\rmid}$ with $\rmid>\rpro$ then $\sigA_P(\tl')\restr\rpro = \sigA_P(\tl)\restr{\rpro}$ and again the invariant is preserved. In other cases \eve immediately wins.

Now consider the case of \Dr in which $\tl(\epsilon)=\spri{\rpro}$. It means that $\ordA_{\rpro}$ is a~successor ordinal greater than $0$. In particular also $\ordA'_{\rpro'}>0$. First, if $\tr(\epsilon)=\spla{P}$ then \adam cannot decrease $\sigA_P(\tr)$ in his move and the invariant must be preserved. If $\tr(\epsilon)=\spla{\bar{P}}$ then the preservation guarantee (see Definition~\ref{def:preservation}) makes sure that no matter what \eve does, the invariant is still preserved. Notice that using Item~\ref{it:sig-min} of Lemma~\ref{lem:sig-exist} we know that at least one of the subtrees $\tr_\dL$ or $\tr_\dR$ must satisfy the preservation guarantee, as
\[\sigA_P(\tl)\restr\rpro\lexeq\sigA_P(\tr)\restr\rpro=\max(\sigA_P(\tr_\dL),\sigA_P(\tr_\dR))\restr\rpro=\max\big(\sigA_P(\tr_\dL)\restr\rpro,\sigA_P(\tr_\dR)\restr\rpro\big).\]
This means that \eve has at least one choice here. Similarly as before, if $\tr(\epsilon)=\spri{\rmid}$ with $\rmid>\rpro$ then the invariant stays the same. What remains is to exclude the possibility that \adam immediately wins in $\downr$. If $\tr(\epsilon)=\spri{\rpro}$ then we would move to the case \Db. If $\tr(\epsilon)=\spri{\rmid}$ with $\rmid<\rpro$ or $\tr(\epsilon)=\sneg$ then $\ordA'_{\rpro}=0$, contradicting our assumptions. These are the only two possible cases.

\paragraph*{The case of \texorpdfstring{\adam}{A}}

Now consider the case of $\adam$, his quasi\=/strategy $\quas{\adam}$, and a~position $(\tl,\tr,\rpro)$ that violates~\eqref{eq:winner-in-C}, what means that $\sigA_P(\tl)\restr{\rpro}\lexge\sigA_P(\tr)\restr{\rpro}$. Notice that the strict inequality here implies that $\sigA_P(\tr)\neq\infty$ and $P$ wins $\Gg(\tr)$. We will consider the successive round of $\Cc_P$ in which $\adam$ plays according to $\quas{\adam}$. 

First consider the step $\EL$ in which \eve may claim that $P$ loses $\Gg(\tr)$ but then she loses, as $P$ wins $\Gg(\tr)$.

In the step $\AL$, if \adam decides to declare that $P$ loses $\Gg(\tl)$ then he wins, otherwise $P$ wins $\Gg(\tl)$ as well.

Let $\sigA_P(\tl)=(\ordA_{\rmin'},\ldots,\ordA_{\rmax'})$.

In \EI, if \eve decides to choose a~new number $\rpro'<\rpro$ then the invariant is preserved, as
\[\sigA_P(\tl')\restr{\rpro'}=(\ordA_{\rmin'},\ldots,\ordA_{\rpro'}+1)\lexge(\ordA_{\rmin'},\ldots,\ordA_{\rpro'})\lexgeq\sigA_P(\tr)\restr{\rpro'}.\]

Similarly in \AI, if \adam decides to choose a~new number $\rpro'<\rpro$ according to his quasi\=/strategy then he makes it in such a~way to preserve the invariant~\eqref{eq:winner-in-C}.

If the round did not end so far, we know that $\sigA_P(\tr)\restr\rpro=(\ordA_{\rmin'},\ldots,\ordA_{\rpro-2},\ordA'_{\rpro})$ where all the coordinates except the last are equal and $\ordA_{\rpro}>\ordA'_{\rpro}$. Similarly as in the case of \eve, if we take the step \Db then the invariant is preserved.

Consider the step \Dl and the choices in $\downl$. If $\tl(\epsilon)=\spla{P}$ then similarly as before \eve can only increase $\sigA_P(\tl)$ preserving the invariant. \adam will not violate the invariant in the case of $\tl(\epsilon)=\spla{\bar{P}}$ because of the preservation guarantees, and as Item~\ref{it:sig-min} of Lemma~\ref{lem:sig-exist} implies, at least one choice is left for \adam in that case. If $\tl(\epsilon)=\spri{\rmin}$ with $\rmid>\rpro$ then the invariant is clearly preserved. Since $\ordA_{\rpro}>\ordA'_{\rpro}\geq 0$, \eve cannot immediately win in $\downl$.

Now consider the last case of step \Dr. As before, the choices of \eve, \adam, or the move over a~priority greater than $\rpro$ will not violate the invariant. The remaining case of an immediate win of \adam is also allowed now.

\subsection{Global properties}

We are now in position to prove that the existence of the quasi\=/strategies $\quas{\eve}$ and $\quas{\adam}$ exclude the two cases: that~\eqref{eq:winner-in-C} holds but \adam has a positional winning strategy; and that~\eqref{eq:winner-in-C} is violated and \eve has a positional winning strategy. To achieve that, we will study the properties of the infinite plays consistent with $\quas{\eve}$ and $\quas{\adam}$.

Notice that each play of $\Cc_P$ can modify the value of $\rpro$ only bounded number of times. Moreover, because of the conditions in steps \Db, \Dl, and \Dr we obtain the following fact.

\begin{fact}
\label{ft:form-plays}
If $\plyA$ is an infinite play of $\Cc_P$ then exactly one of the following three cases holds:
\begin{itemize}
\item $\plyA$ makes infinitely many \Db steps,
\item from some point on $\plyA$ makes only \Dl steps,
\item from some point on $\plyA$ makes only \Dr steps.
\end{itemize}
\end{fact}

Observe that each step of $\Cc_P$ of the form $\downl$ or $\downr$ (if it doesn't mean an~immediate win) simulates in fact a~round of the game $\Gg(\tl)$ and $\Gg(\tr)$ respectively. Moreover, the quasi\=/strategies of the players \eve and \adam simulate optimal strategies of $P$ in these rounds respectively. Thus, $P$ must win these plays, as expressed by the following lemma.

\begin{lemma}
\label{lem:honest-optimal}
Consider an infinite play consistent with the quasi-strategy of one of the players ($\quas{\eve}$ or $\quas{\adam}$). Then the play only finitely many times makes the \Db step.

Moreover, if the play follows $\quas{\eve}$ and from some point on makes only \Dl steps then it is winning for \eve. Similarly, if the play follows $\quas{\adam}$ and from some point on makes only \Dr steps then it is winning for \adam.
\end{lemma}

\begin{proof}
First take an infinite play $\plyA$ of the quasi-strategy $\quas{\eve}$ of \eve starting from a~position $(\tl_0,\tr_0,\rpro_0)$. The subtrees of the~tree $\tl_0$ seen during $\plyA$ follow a~(possibly finite) play $\infA$ of an optimal strategy of $P$ in $\Gg(\tl_0)$. Since $\plyA$ is infinite, \eve does not declare that $P$ loses $\Gg(\tr_0)$. Therefore, $\sigA_P(\tr_0)\lex\infty$ and by~\eqref{eq:winner-in-C} we know that also $\sigA_P(\tl_0)\lex\infty$ which implies that $\tl_0\in\Wng{P}{\rmin}{\rmax}$. By Lemma~\ref{lem:sig-opt-win} the play $\infA$ in $\Gg(\tl_0)$ follows a~winning strategy of $P$.

We will show that the step \Db occurs only finitely many times in $\plyA$. Assume contrarily and let $\rpro$ be the minimal number $\rpro$ that occurs in the play $\plyA$. Using the above assumptions, we know that $\rpro$ is the minimal priority that is seen in the tree $\tl_0$ infinitely many times on $\infA$. Therefore, the simulated play $\infA$ in $\Gg(\tl_0)$ is infinite and losing for $P$, which is a~contradiction.

Therefore, by Fact~\ref{ft:form-plays} the play $\plyA$ either makes from some point on only \Dl steps, or from some point on only \Dr steps. In the first case it follows the play $\infA$ of $\Gg(\tl_0)$ that is winning for $P$. By the choice of priorities in the step $\downl$ we know that \eve wins~$\plyA$.

The case of the quasi-strategy $\quas{\adam}$ of \adam is entirely dual: we use the assumption that~\eqref{eq:winner-in-C} is violated to know that $\sigA_P(\tr)\lex\infty$ so $\tr_0\in\Wng{P}{\rmin}{\rmax}$. Moreover, the choice of priorities in the step $\downr$ implies that if $\plyA$ makes $\downr$ infinitely many times then \adam wins $\plyA$.
\end{proof}

Now we move to the proof of the two cases we need to exclude.

\myPar{The case of \texorpdfstring{\eve}{E}} Assume that a~position $\big(\tl,\tr,\rpro\big)$ of $\Cc_P$ satisfies~\eqref{eq:winner-in-C} but \adam has a~positional winning strategy $\Sigma_\adam$ from that position. We will prove that such a~case is not possible.

By Lemma~\ref{lem:honest-safety}, the positional quasi-strategy $\quas{\eve}$ always stays within positions satisfying the~invariant~\eqref{eq:winner-in-C}. Moreover, the quasi\=/strategy never reaches a~position that is immediately losing for \eve. Similarly, $\Sigma_\adam$ never reaches a~position that is immediately losing for \adam. Thus, all the plays consistent with both $\quas{\eve}$ and $\Sigma_\adam$ must be infinite.

First notice that the values of $\rpro$ are non\=/increasing during the plays of $\Cc_P$ and therefore, there exists a~position that belongs to both $\quas{\eve}$ and $\Sigma_\adam$ such that the value $\rpro$ stays constant during all the plays from that position. Without loss of generality we can assume that this is our starting position. 

We can now proceed inductively in the tree obtained by unravelling the intersection of $\quas{\eve}$ and $\Sigma_\adam$: whenever the currently considered subtree contains anywhere a~\Db step, we change the initial position to the result of that step. By Lemma~\ref{lem:honest-optimal} no play consistent with $\quas{\eve}$ takes the \Db step infinitely many times. Therefore, our inductive procedure has to stop at some point with no \Db steps in the current subtree. Without loss of generality we can assume that the initial position $(\tl_0,\tr_0,\rpro_0)$ is the last position from the procedure. We know that the plays consistent with both $\quas{\eve}$ and $\Sigma_\adam$ never take the \Db step nor modify $\rpro=\rpro_0$.

The structure of $\Cc_P$ guarantees that since the step \Db is not allowed, each play consistent with both $\quas{\eve}$ and $\Sigma_\adam$ takes only \Dl steps or takes only \Dr steps. Lemma~\ref{lem:honest-optimal} implies that in the former case the play would be winning for \eve, contradicting the assumption that $\Sigma_\adam$ is winning. Thus, all the considered plays take only \Dr steps. In particular $\tl=\tl_0$ is constant.

The intersection of $\Sigma_\adam$ and $\quas{\eve}$ induces a~partial strategy $\Sigma_P$ of $P$ in $\Gg(\tr_0)$---$\Sigma_P$ is partial because it does not contain positions that cannot be reached by following $\quas{\eve}$ because of the preservation guarantees, see Definition~\ref{def:preservation}. The subtrees $\tr'$ of $\tr_0$ in such unreachable positions satisfy $\sigA_P(\tr')\restr{\rpro_0}\lex\sigA_P(\tl_0)\restr{\rpro_0}$ by the definition of $\quas{\eve}$. In the positions on which~$\Sigma_P$ is defined it never visits a~priority $\rmid$ with $\rmid\leq \rpro_0$ nor a~node labelled $\sneg$ because such a~move is immediately losing for \adam in $\downr$. Because of the choice of priorities in $\downr$ and since $\Sigma_\adam$ is winning, $\Sigma_P$ is winning for $P$ on infinite plays.

Notice that since we take only \Dr steps, $\tl_0$ must be of the form $\spri{\rpro_0}(\tl_0')$. Therefore, $\sigA_P(\tl_0)\restr{\rpro_0}=(\ordA_{\rmin'},\ldots,\ordA_{\rpro_0}+1)$ for $(\ordA_{\rmin'},\ldots,\ordA_{\rpro_0})\eqdef\sigA_P(\tl_0')\restr{\rpro_0}$. It means that whenever the partial strategy $\Sigma_P$ cannot reach a~position with a~subtree $\tr'$, we know that in fact $\sigA_P(\tr')\restr{\rpro_0}\lexeq (\ordA_{\rmin'},\ldots,\ordA_{\rpro_0})$. The following lemma says that the existence of such a~partial strategy $\Sigma_P$ witnesses the inequality $\sigA_P(\tr_0)\restr{\rpro_0}\lexeq(\ordA_{\rmin'},\ldots,\ordA_{\rpro_0})$. By the definition of the ordinals $\ordA_{\rmid}$ we know that $(\ordA_{\rmin'},\ldots,\ordA_{\rpro_0})\lex\sigA_P(\tl_0)\restr{\rpro_0}$, what contradicts~\eqref{eq:winner-in-C} for $(\tl_0,\tr_0,\rpro_0)$.


\myPar{The case of \texorpdfstring{\adam}{A}} Now we move to the dual case---we assume that the initial position $\big(\tl_0,\tr_0,\rpro_0\big)$ violates~\eqref{eq:winner-in-C} but \eve has a~positional winning strategy $\Sigma_\eve$. Similarly as in the previous case we can assume (by modifying the initial position $(\tl_0,\tr_0,\rpro_0)$) that in the plays consistent with both $\quas{\adam}$ and $\Sigma_\eve$ only the \Dl step is done. In that case both $\tr=\tr_0$ and $\rpro=\rpro_0$ are constant.

Similarly as before, the intersection of $\Sigma_\eve$ and $\quas{\adam}$ induces a~partial strategy $\Sigma_P$ of $P$ in~$\Gg(\tl_0)$ that is defined until a~position with a~subtree $\tl'$ satisfying $\sigA_P(\tl')\restr{\rpro_0}\lexeq\sigA_P(\tr_0)\restr{\rpro_0}$ is visited (the preservation guarantees). In the same way as previously, the partial strategy~$\Sigma_P$ never visits a~position of priority $\rmid\leq\rpro_0$ nor the symbol $\sneg$ and it is winning on infinite plays. Define $(\ordA_{\rmin'},\ldots,\ordA_{\rpro_0})\eqdef\sigA_P(\tr_0)\restr{\rpro_0}$ and apply Lemma~\ref{lem:sig-subopt} to notice that $\sigA_P(\tl_0)\restr{\rpro_0}\lexeq(\ordA_{\rmin'},\ldots,\ordA_{\rpro_0})=\sigA_P(\tr_0)\restr{\rpro_0}$ what implies that~\eqref{eq:winner-in-C} holds in $(\tl_0,\tr_0,\rpro_0)$, which is a~contradiction.

This way we have concluded the proof of Claim~\ref{cl:winner-in-C} and therefore the whole reasoning of this paper is complete.

\section{Why \texorpdfstring{$\sneg$}{\tilde} is needed?}
\label{ap:neg-needed}

In this section we provide a~simple example explaining why the swapping symbol $\sneg$ is needed to properly construct the reduction $c_P$.

\begin{proposition}
\label{pro:no-neg-no-red}
There is no continuous function $\fun{c'}{\big(\trees_{\Alp{0}{\rmax}}\big)^2}{\trees_{\Alp{0}{\rmax}}}$ such that
\[c'(t_\dL,t_\dR)\in\Win{\pI}{0}{\rmax}\ \text{ if and only if }\ \sigA_{\pI}(t_\dL)\lexeq \sigA_{\pI}(t_\dR).\]
\end{proposition}

\begin{proof}
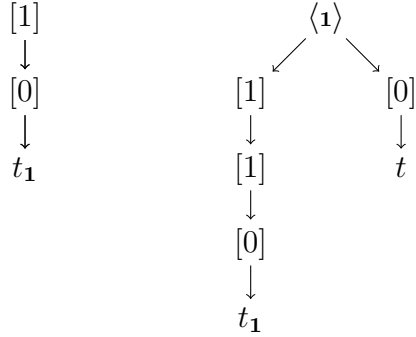
\begin{figure}
\centering
\begin{tikzpicture}
\node[letter] (l0) at ( 0.0, -0) {$\spri{1}$};
\node[bstate] (l1) at ( 0.0, -1) {$\spri{0}$};
\node[bstate] (l2) at ( 0.0, -2) {$t_{\pI}$};
\draw[transs] (l0) -- (l1);
\draw[transs] (l1) -- (l2);

\node[letter] (r0) at ( 4.0, -0) {$\spla{\pI}$};
\node[bstate] (rl1) at ( 3.0, -1) {$\spri{1}$};
\node[bstate] (rl2) at ( 3.0, -2) {$\spri{1}$};
\node[bstate] (rl3) at ( 3.0, -3) {$\spri{0}$};
\node[bstate] (rl4) at ( 3.0, -4) {$t_{\pI}$};

\node[bstate] (rr1) at ( 5.0, -1) {$\spri{0}$};
\node[bstate] (rr2) at ( 5.0, -2) {$t$};

\draw[transs] (l0) -- (l1);
\draw[transs] (l1) -- (l2);

\draw[transs] (r0) -- (rl1);
\draw[transs] (rl1) -- (rl2);
\draw[transs] (rl2) -- (rl3);
\draw[transs] (rl3) -- (rl4);

\draw[transs] (r0) -- (rr1);
\draw[transs] (rr1) -- (rr2);

\end{tikzpicture}
\caption{The pair of trees being the result of the reduction $r(t)$ from Proposition~\ref{pro:no-neg-no-red}.}
\label{fig:r-reduction}
\end{figure}

Assume that such a~function $c'$ exists. Fix a~tree $t_{\pI}\in\Win{\pI}{0}{\rmax}$ and consider a~function $\fun{r}{\trees_{\Alp{0}{\rmax}}}{\big(\trees_{\Alp{0}{\rmax}}\big)^2}$ defined as follows:
\[r(t)=\Big(\spri{1}(\spri{0}(t_{\pI})),\ \spla{\pI}(\spri{1}(\spri{1}(\spri{0}(t_{\pI})), \spri{0}(t))\Big),\]
see Figure~\ref{fig:r-reduction} for a~pictorial representation of the two trees. Clearly $r$ is continuous.

Let $t\in\trees_{\Alp{0}{\rmax}}$ and $r(t)=(t_\dL,t_\dR)$. Notice that $\sigA_{\pI}(t_\dL)=(1,0,\ldots)$. The value $\sigA_{\pI}(t_\dR)$ is either $(0,0,\ldots)$ if $t\in\Win{\pI}{0}{\rmax}$ or $(2,0,\ldots)$ otherwise. Therefore, $c'\big(r(t)\big)\in\Win{\pI}{0}{\rmax}$ if and only if $\sigA_{\pI}(t_\dL)\lexeq \sigA_{\pI}(t_\dR)$ if and only if $t\notin\Win{\pI}{0}{\rmax}$ if and only if $t\in\Win{\pII}{0}{\rmax}$. Thus, $\fun{c'\circ r}{\trees_{\Alp{0}{\rmax}}}{\trees_{\Alp{0}{\rmax}}}$ is a~continuous reduction of $\Win{\pII}{0}{\rmax}$ to $\Win{\pI}{0}{\rmax}$. This is a~contradiction with~\cite[Lemma~1]{niwinski_strict} (the assumption of contractivity is redundant there by Lemma~2 from the same paper).
\end{proof}

\begin{remark}
The same argument shows lack of the reduction if we insisted on strict inequality between the $\pI$\=/signatures $\sigA_{\pI}(t_\dL)\lex \sigA_{\pI}(t_\dR)$.

By duality, there is also no reduction in the case of $\rmin=1$ and $P=\pII$.
\end{remark}

\comment{
\section{Hardness of the language}
\label{ap:hardness}

In this section we provide remaining details from the proof of Proposition~\ref{pro:W-reduction}. Recall that the function $\fun{f}{\trees_{\Ang{\rmin}{\rmax}}}{\trees_{\Angp{\rmin}{\rmax}}}$ was defined recursively as:
\begin{align*}
f\big(\spla{P}(t_\dL,t_\dR)\big)&=\splp{P}\big(f(t_\dL),f(t_\dR),f\big(c_P(t_\dL,t_\dR)\big)\big),\\
f\big(\sneg(t)\big)&=\sneg\big(f(t)\big),\\
f\big(\spri{\rmid}(t)\big)&=\spri{\rmid}\big(f(t)\big).
\end{align*}

First notice that $f$ is well\=/defined, as it always produces a~new part of a~tree before recursively applying itself again. As noted in the main body of the article, $\shave\big(f(t)\big)=t$. Additionally, the value of $f(t)$ on a~fixed-depth depends only on a~bounded number of iterations of $c_P$ and since a~composition of continuous functions is continuous, this fixed-depth part of $f(t)$ depends only on a~fixed-depth part of $t$. Therefore, $f$ is a~continuous function. Clearly, as $c_P$ maps well\=/formed trees to well\=/formed trees and being well\=/formed depends only on $\shave$ of the subtrees, $f$ also maps well\=/formed trees to well\=/formed trees.

Now, consider a~well\=/formed tree $t\in\trees_{\Ang{\rmin}{\rmax}}$. The ``only if'' part of the equivalence was already proved in the main body of the article. We need to prove the opposite ``if'' part.

\begin{replemma}{lem:construct-run}
\lemConstructRun{\comment}
\end{replemma}

\begin{proof}
The construction of $\rho$ is inductive: we go down the tree $r$. Consider a~subtree $r'=r\restr\finA$ of $r$ under a~current node $\finA\in\dom\big(r\big)$. The second coordinate $P$ of $\rho(\finA)$ (i.e.~$\pI$ or $\pII$) depends on whether $\pI$ or $\pII$ wins in $\Gg(\shave(r'))$. In the case the second coordinate of $\rho(\finA)$ is $P$ and $r'$ is of the form $\splp{P}\big(f(t_\dL),f(t_\dR),f(c_P(t_\dL,t_\dR)\big)$ then let the third coordinate $d$ be $\dL$ if $\pI$ wins $\Gg\big(r\restr\finA\cdot 2\big)$ and $\dR$ otherwise. In the remaining cases let the third coordinate $d$ be $\dD$.

First notice that together with the transitions of $\Uu$ the above definition in fact produces a~run, preserving the invariant that the current second coordinate is $P$ if and only if $P$ wins in $\Gg(\shave(r'))$ of the current subtree~$r'$. Lemma~\ref{lem:second-to-other} shows that such a~run is unique.
\end{proof}
}

\section{Relations to the rank method}
\label{ap:ranks}

This section provides an explanation of relations between $P$\=/signatures together with Lemma~\ref{lem:reduction}; and rank\=/comparison theorems.

$P$\=/signatures can be seen as a stratification of the winning set $\Wng{P}{\rmin}{\rmax}$ into a hierarchy of length $\omega_1$ (the first uncountable ordinal). This resembles the $\api{1}$\=/rank method (see for instance~\cite[Theorem~34.4]{kechris_descriptive}) which allows to stratify any set $X\in\api{1}$ by a \emph{rank} assignment: $X\ni x\mapsto\rank(x)<\omega_1$. The crucial property of that assignment is that the rank\=/order relation is of the same complexity as the original set:
\[\big\{(x,y)\mid \rank(x)\leq\rank(y)\big\}\in\asigma{1}.\]
This method was further extended by Lyapunov~\cite{lyapunov_classification} to the class of $\Rr$\=/sets.

The class of $\api{1}$ sets coincides with $\Rr$\=/sets on the first level of the $\Rr$\=/hierarchy. Moreover, the set $\Win{\pI}{1}{2}$ is topologically complete for that class. As one can show, in that basic case all the three notions coincide: $\rank(t)$, Lyapunov's $\Rr$\=/rank of $t$, and the~$\pI$\=/signature of $t$.

When we move higher in the index hierarchy, the correspondence between the winning sets $\Win{\pI}{\rmin}{\rmax}$ and the hierarchy of $\Rr$\=/sets is preserved, as expressed by the following theorem.

\begin{theorem}[{Michalewski et al~\cite{michalewski_measure_inf}}]
The set $\Win{\pI}{1}{\rmax+1}$ is topologically complete for the $\rmax$th level of the hierarchy of $\Rr$\=/sets.
\end{theorem}

Moreover, Lyapunov~\cite{lyapunov_classification} claims the following theorem.
\begin{theorem}[Lyapunov~\cite{lyapunov_classification} (without a proof), see also~\cite{konovei_operations}]
For every $\rmax=1,\ldots$ the $\Rr$\=/rank order relation for $\Rr$\=/sets on the $\rmax$th level of the hierarchy lies on the same $\rmax$th level of the hierarchy.
\end{theorem}

Thus, as $\Win{\pI}{1}{\rmax+1}$ is complete for this level, the above theorem gives a~reduction between the $\Rr$\=/rank order on $\Win{\pI}{1}{\rmax+1}$ and the set $\Win{\pI}{1}{\rmax+1}$, similarly as in Lemma~\ref{lem:reduction}.

The only problem is that the notion of $\Rr$\=/rank cannot be used to construct optimal winning strategies in the sense of Lemma~\ref{lem:sig-opt-win}. The reason for that is the difference between $\Rr$\=/ranks and $\pI$\=/signatures: $\pI$\=/signatures count the number of occurrences of $\pI$\=/losing priorities until the first lower priority is seen; while $\Rr$\=/rank of a tree counts these occurrences uniformly (i.e.~it is the supremum of signatures of the subtrees inside an optimal strategy).

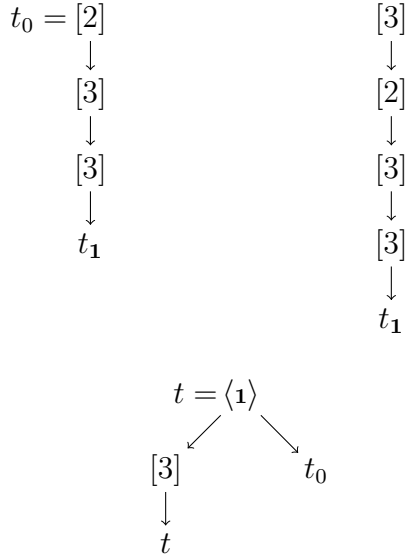
\begin{figure}
\centering
\begin{tikzpicture}
\node[toL] at (-0.3, 0) {$t_0=$};
\node[letter] (l0) at ( 0.0, -0) {$\spri{2}$};
\node[bstate] (l1) at ( 0.0, -1) {$\spri{3}$};
\node[bstate] (l2) at ( 0.0, -2) {$\spri{3}$};
\node[bstate] (l3) at ( 0.0, -3) {$t_{\pI}$};

\draw[transs] (l0) -- (l1);
\draw[transs] (l1) -- (l2);
\draw[transs] (l2) -- (l3);

\node[letter] (r0) at ( 4.0, -0) {$\spri{3}$};
\node[bstate] (r1) at ( 4.0, -1) {$\spri{2}$};
\node[bstate] (r2) at ( 4.0, -2) {$\spri{3}$};
\node[bstate] (r3) at ( 4.0, -3) {$\spri{3}$};
\node[bstate] (r4) at ( 4.0, -4) {$t_{\pI}$};

\draw[transs] (r0) -- (r1);
\draw[transs] (r1) -- (r2);
\draw[transs] (r2) -- (r3);
\draw[transs] (r3) -- (r4);

\node[toL] at (1.7, -5) {$t=$};
\node[letter] (r0) at ( 2.0, -5) {$\spla{\pI}$};
\node[bstate] (rl1) at ( 1.0, -6) {$\spri{3}$};
\node[bstate] (rl2) at ( 1.0, -7) {$t$};

\node[bstate] (rr1) at ( 3.0, -6) {$t_0$};

\draw[transs] (r0) -- (rl1);
\draw[transs] (rl1) -- (rl2);

\draw[transs] (r0) -- (rr1);
\end{tikzpicture}
\caption{An illustration of differences between $\pI$\=/signatures and the $\Rr$\=/rank.}
\label{fig:two-trees}
\end{figure}

\begin{example}
Let $t_\pI$ satisfy $t_\pI=\spri{2}(t_\pI)\in\Win{\pI}{1}{3}$. The $\Rr$\=/rank of the trees $t_0$ and $\spri{3}(t_0)$ depicted in the upper part of Figure~\ref{fig:two-trees} are the same (equal $(0,2)$ assuming a~tuple notation as for signatures). Therefore, the strategy $\Sigma$ of $\pI$ in $\Gg(t)$ that moves always to the left is optimal with respect to the $\Rr$\=/ranks, but losing.

The $\pI$\=/signatures of $t_0$ and $\spri{3}(t_0)$ are respectively $(0,0)$ and $(0,1)$. Thus, the only $\sigA$\=/optimal strategy in $\Gg(t)$ is the one that moves immediately to the right, reaches $t_0$, and wins.
\end{example}

\section{Conclusions}
\label{sec:conclusions}

The main result of this work is the construction of the languages $\lan{P}{\rmin}{\rmax}$ that solve the question of index bounds for unambiguous languages. Although the construction is not direct and relies heavily on an involved theory of signatures, these complications seem to be unavoidable when one wants to recognise languages like $\W{\rmin}{\rmax}$ in an unambiguous way.

The definition of signatures given in the paper seems to be the canonical one, as witnessed by the point\=/wise minimality from Lemma~\ref{lem:sig-exist}. The previous ways of using signatures were mainly focused on their monotonicity and well\=/foundedness, thus it was enough to assume inequalities in the invariants of Lemma~\ref{lem:sig-exist}. Here we are interested in comparing their actual values, therefore we insist on preserving these values via equalities.

Another novelty of this paper is the game $\Cc_P$ that reduces the order of the signatures in two $(\rmin,\rmax)$\=/parity games (with the swap $\sneg$) to another $(\rmin,\rmax)$\=/parity game (also with $\sneg$). Both, a~stronger variant of signatures and the game $\Cc_P$ might be relevant for a~general study of parity games.

{\bf Acknowledgements}. The author would like to thank Szczepan Hummel, Damian Niwi{\'n}ski, and Igor Walukiewicz for fruitful and inspiring discussions on the topic. Moreover, the author is grateful to Bartek Klin, Kamila {\L}yczek, Filip Murlak, Grzegorz Rz{\k a}ca, and the anonymous referees for a~number of editorial suggestions about the paper.

\bibliography{mskrzypczak}

\end{document}